\documentclass[english,13pt]{article}
\usepackage[T1]{fontenc}
\usepackage[latin1]{inputenc}
\usepackage{geometry}
\usepackage{float}
\usepackage{mathrsfs}
\usepackage{amsmath}
\usepackage{amssymb}
\usepackage{graphicx}
\usepackage{hyperref}
\usepackage[color=yellow]{todonotes}
\hypersetup{
    colorlinks=true,
    linkcolor=blue,
    filecolor=magenta,      
    urlcolor=cyan,
    citecolor = blue,
}

\usepackage{yhmath}
\makeatletter

\usepackage{bm}
\numberwithin{equation}{section}

\usepackage{algorithm}
 \usepackage{algorithmicx}
\floatstyle{ruled}
\newfloat{algorithm}{tbp}{loa}
\providecommand{\algorithmname}{Algorithm}
\floatname{algorithm}{\protect\algorithmname}

\usepackage{algorithm}
 \usepackage{algorithmicx}

\usepackage{amsthm}

\usepackage{mathrsfs}
\usepackage{amssymb}

\usepackage{amsfonts}

\usepackage{epsfig}

\usepackage{bm}

\usepackage{mathrsfs}

\usepackage{enumerate}

\@ifundefined{definecolor}{\@ifundefined{definecolor}
 {\@ifundefined{definecolor}
 {\usepackage{color}}{}
}{}
}{}

\usepackage{subfig}\usepackage[all]{xy}

\newtheorem{theorem}{Theorem}[section]

\newcounter{hypA}
\newenvironment{hypA}{\refstepcounter{hypA}\begin{itemize}
  \item[({\bf A\arabic{hypA}})]}{\end{itemize}}
\newcounter{hypB}

\newcounter{hypD}

\usepackage{babel}\date{}

\makeatother

\begin{document}

\begin{center}

{\Large \textbf{Unbiased Parameter Estimation for Partially Observed Diffusions}}

\vspace{0.5cm}

ELSIDDIG AWADELKARIM, AJAY JASRA \& HAMZA RUZAYQAT

{\footnotesize Applied Mathematics and Computational Science Program, \\ Computer, Electrical and Mathematical Sciences and Engineering Division, \\ King Abdullah University of Science and Technology, Thuwal, 23955-6900, KSA.} \\
{\footnotesize E-Mail:\,} \texttt{\emph{\footnotesize elsiddigawadelkarim.elsiddig@kaust.edu.sa}}, \\
\texttt{\emph{\footnotesize ajay.jasra@kaust.edu.sa}},  \texttt{\emph{\footnotesize hamza.ruzayqat@kaust.edu.sa}}

\begin{abstract}
In this article we consider the estimation of static parameters for partially observed diffusion process with discrete-time observations over a fixed time interval. In particular,  we assume that one must time-discretize the partially observed diffusion process and work with the model with bias and consider maximizing the resulting log-likelihood. Using a novel double randomization scheme, based upon Markovian stochastic approximation we develop a new method to unbiasedly estimate the static parameters, that is, to obtain the maximum likelihood estimator with no time discretization bias.
Under assumptions we prove that our estimator is unbiased and investigate the method in several numerical examples, showing that it can empirically out-perform existing unbiased methodology.
\noindent\textbf{Corresponding author}: Hamza Ruzayqat. E-mail:
\href{mailto:hamza.ruzayqat@kaust.edu.sa}{hamza.ruzayqat@kaust.edu.sa}
\end{abstract}

\end{center}

\section{Introduction}

Let $(\mathsf{X},\mathcal{X})$ be a measurable space, $\Theta\subseteq\mathbb{R}^{d_{\theta}}$ and introduce the family of probability measures $\{\pi_{\theta}\}_{\theta\in\Theta}$,
that is, for each $\theta\in\Theta$, $\pi_{\theta}$ is a probability measure on $(\mathsf{X},\mathcal{X})$. Let $H:\Theta\times\mathsf{X}\rightarrow\mathbb{R}^{d_h}$ be a measurable mapping such that for each $\theta\in\Theta$, $h(\theta) := \int_{\mathsf{X}}H(\theta,x)\pi_{\theta}(dx)$ is finite. The objective is to solve $h(\theta)=0$. This problem appears routinely in many applications such as maximum likelihood estimation of parameters associated to partially observed diffusion processes, where $\pi_{\theta}$ represents the posterior measure of the partially observed process for a fixed $\theta$ and $H$ is an expression associated to the gradient of the log-likelihood obtained by Fisher's identity; more details are given later on - see also \cite{ub_grad}. Other applications include parameter estimation for Bayesian inverse problems (e.g.~\cite{disc_model,par_uq}) although this is not the focus of this article. The afore-mentioned problems have a wide variety of applications in statistics, finance and engineering; see \cite{cappe,golightly} for example.

We assume that in practice, one requires a discretization (approximation) associated to the probability $\pi_{\theta}$ that is controlled by a scalar parameter $l\in\mathbb{N}_0=\mathbb{N}\cup\{0\}$. We denote such a probability as $\pi_{\theta}^l$ 
on a measurable space $(\mathsf{X}^l,\mathcal{X}^l)$ (which may be different from $(\mathsf{X},\mathcal{X})$ but need not be) and suppose that there is a corresponding functional $H_l:\Theta\times\mathsf{X}^l\rightarrow\mathbb{R}^{d_h}$ writing
$h_l(\theta) := \int_{\mathsf{X}^l}H_l(\theta,x)\pi_{\theta}^l(dx)$. 
We will assume explicitly that $\lim_{l\rightarrow\infty}h_l(\theta)=h(\theta)$ for every $\theta\in\Theta$ and again this is made precise in the next section.
In the context of diffusion processes, which is the example considered in this article, $l$ will control the level of time discretization  and this example will be detailed in the next section. The objective article is still to solve $h(\theta)=0$, despite the fact that
one can only work with $h_l(\theta)$. The solution we obtain is unbiased, in the sense that the expectation of the estimator is equal to the (assumed unique) solution of $h(\theta)=0$.

One of the main methods that can be used to solve $h_l(\theta)=0$ is based upon the well-known stochastic approximation (SA) method and its many variants; see for instance \cite{andr3,andr,Der,frika,robbins}. This is an iterative scheme that updates parameter estimates using often (but not-necessarily) unbiased estimates of $h_l(\theta)$, which are independent at each update. The main complication in our context (partially observed diffusions) is that unbiased estimation w.r.t.~$\pi_{\theta}^l$ is often very challenging as one is working with a discrete-time state-space model. The posteriors $\pi_{\theta}^l$ are notoriously challenging to sample from and there have been a large number of methods designed to do exactly this; see for instance \cite{andrieu,graham}. Such methods are based upon Markov chain Monte Carlo (MCMC) and at least in the context of \cite{andrieu}, SA has been combined with using MCMC, to yield a Markovian update in the SA method to yield a type of Markovian SA (MSA) method. This can be far more efficient than obtaining an independent (and unbiased) estimator of $h_l(\theta)$ as it typically needs only one or few updates from a single Markov kernel; details are given in the next section.
Note that the MSA method will only give solutions of $h_l(\theta)=0$ and so the resulting solution is likely to have a bias.

In this article based upon the randomization methods in \cite{rhee}, extended to so-called doubly randomized schemes in \cite{ub_grad,disc_model,par_uq,ub_pf, sfs}, we develop a new unbiased parameter estimation method for partially observed diffusion processes that uses the MSA method that is mentioned above. By unbiased, we mean a method than can obtain the solution of $h(\theta)=0$, despite the fact that one only works with $\pi_{\theta}^l$ and $h_l(\theta)$. The approach consists, roughly, of randomizing over the level of discretization $l$ and then running a MSA method at two consecutive levels of time discretization for a random number of iterations of an MSA algorithm. Using this approach one can show, under assumptions, that one can solve
$h(\theta)=0$. Key to this methodology is being able to run couplings of MCMC kernels that target $\pi_{\theta}^l$ and $\pi_{\theta}^{l-1}$ where the couplings are `suitably good'. Such methods have been developed in \cite{ub_grad} and are utilized here. We note that the work here is not the first to provide unbiased parameter estimates associated to partially observed diffusions with discrete-time observations. Some of the first methods are based upon exact simulations of diffusions (e.g.~\cite{beskos}), which are rather elegant when they can be applied; we assume that this is not possible and this is a reasonably wide-class of diffusion processes.
The approach in \cite{ub_grad} is more general than that of \cite{beskos}, using a methodology to unbiasedly estimate $h(\theta)$ at each time step of the SA method and returning unbiased parameter estimates. Note that the estimator that the authors is unbiased in terms of discretization error, but has a bias from the SA iterations; this is a bias that we will remove.
The method in \cite{ub_grad} is in contrast to the approach in this article where we focus explicitly on unbiasedly estimating the parameters themselves. The advantage of the latter approach, is that often the cost-per-iteration of the MSA method will be less than having to unbiasedly estimate $h(\theta)$ and hence that one has a faster computational algorithm for parameter estimation. It should be noted, however, that unbiased estimation of the gradient of the log-likelihood (i.e.~$h(\theta)$) is of independent interest in itself and can be used inside MCMC methods and so on, whereas our approach is limited to exactly parameter estimation. 

To summarize the contributions of this article are as follows:
\begin{enumerate}
\item{To provide a new method for unbiased estimation of static parameters of partially observed diffusions associated to discrete-time observations over a fixed interval.}
\item{To prove that (a version of) the estimator is indeed unbiased, under mathematical assumptions.}
\item{To implement the method numerically on several examples and to establish faster empirical performance for parameter estimation than that in \cite{ub_grad}.}
\end{enumerate}
We note that the analysis in 2.~is non-trivial and relies on proving that the MSA method with reprojection (which is detailed later on) converges in our context. We use the main convergence theorem of \cite{andr3} which requires several conditions to be  
verified, associated to drift and minorization conditions of the Markov kernels that are used (amongst other conditions). This is non-trivial as we use a type conditional particle filter \cite{andrieu} that is in turn coupled. We do not prove that our estimator has finite variance although we conjecture exactly how that may be proved.

This article is structured as follows. In Section \ref{sec:method} we detail the mathematical problem and our methodology.
Section \ref{sec:diffusions} outlines the methodology in the context of diffusions and contains within it Section \ref{sec:theory} which gives our main theoretical result on the unbiasedness of the parameter estimates in the case of partially observed diffusion processes.
In Section \ref{sec:numerics} we provide numerical simulations which investigate our methodology.
Appendix \ref{app:proofs} houses all of our technical proofs.

\section{Methodology}\label{sec:method}

\subsection{Problem Formulation}

Recall that $h(\theta) := \int_{\mathsf{X}}H(\theta,x)\pi_{\theta}(dx)$.
The objective is to solve, $h(\theta)=0$, which can be achieved using Markovian stochastic approximation (e.g.~\cite{andr}) as we will now describe. Let $K_{\theta}:\mathsf{X}\times\mathcal{X}\rightarrow[0,1]$ be a Markov kernel, such that for any $\theta\in\Theta$, it admits $\pi_{\theta}$ as an invariant measure and, for each $\theta\in\Theta$, let $\nu_\theta$ be a probability measure on $(\mathsf{X},\mathcal{X})$. Set $\theta_0\in\Theta$ and generate $X_0\sim\nu_{\theta_0}$. Then Markovian stochastic approximation 
iterates at time $\theta\in\mathbb{N}$, as follows:
\begin{enumerate}
\item{Sample $X_n|(\theta_0,x_{0}),\dots,(\theta_{n-1},x_{n-1})$ from $K_{\theta_{n-1}}(x_{n-1},\cdot)$.}
\item{Update:
$$
\theta_n = \theta_{n-1} - \gamma_n H(\theta_{n-1},X_n)
$$
where $\{\gamma_n\}_{n\in\mathbb{N}}$ is a sequence of step-sizes, $\gamma_n>0$, $\sum_{n\in\mathbb{N}}\gamma_n=\infty$, $\sum_{n\in\mathbb{N}}\gamma_n^2<\infty$.}
\end{enumerate}
Typically, to prove convergence one often needs reprojections and ergodicity conditions on $K_{\theta}$, but for now we shall ignore this to simplify the exposition.

We now assume that working directly with $\pi_{\theta}$ is not possible and one can only work with a 
family of approximations $\{\pi_{\theta}^l\}_{l\in\mathbb{N}_0}$, 
on a measurable space $(\mathsf{X}^l,\mathcal{X}^l)$,
$l\in\mathbb{N}_0$. This approximation is such that as $l\rightarrow\infty$ $(\mathsf{X}^l,\mathcal{X}^l)$
and $(\mathsf{X},\mathcal{X})$ will be identical. Now for any $\pi_{\theta}^l-$integrable $\varphi_{\theta}^l:\mathsf{X}^l\rightarrow\mathbb{R}$ 
$$
\pi_{\theta}^l(\varphi^l) := \int_{\mathsf{X}^l}\varphi_{\theta}^l(x^l)\pi_{\theta}^l(dx),
$$
we will assume that there exists a $\varphi_{\theta}:\mathsf{X}\rightarrow\mathbb{R}$ that is $\pi_{\theta}-$integrable and
$$
\lim_{l\rightarrow\infty} \pi_{\theta}^l(\varphi^l) = \pi_{\theta}(\varphi),
$$
and is such that $|\pi_{\theta}^l(\varphi^l)-\pi_{\theta}(\varphi)|\leq|\pi_{\theta}^{l-1}(\varphi^{l-1})-\pi_{\theta}(\varphi)|$. We note that the context here exists in partially
observed diffusions and Bayesian inverse problems; see \cite{ub_grad,disc_model} and the former is described in detail in Section \ref{sec:diffusions}.

A Markovian stochastic approximation (MSA) scheme would work as follows. Let $K_{\theta,l}:\mathsf{X}^l\times\mathcal{X}^l\rightarrow[0,1]$ be a Markov kernel, such that for any $\theta\in\Theta$, it admits $\pi_{\theta}^l$ as an invariant measure and, for each $\theta\in\Theta$, let $\nu_\theta^l$ be a probability measure on $(\mathsf{X},\mathcal{X})$. Then one can run Algorithm \ref{alg:SMA}.  In Algorithm \ref{alg:SMA} we do not specify any stopping rule, although of course in practice one must specify one. We note that, assuming that the zero of the functional $h_l(\theta) = \int_{\mathsf{X}^l}H_l(\theta,x)\pi_{\theta}^l(dx)$ exists, one expects it to be different to
the zero of $h(\theta)$.
 
\begin{algorithm}
\caption{Markovian Stochastic Approximation}
\label{alg:SMA}
\begin{algorithmic}[1]
\State{Set $\theta_0^l\in\Theta$ and generate $X_0\sim\nu_{\theta_0}^l$, $n=1$.}
\State{Sample $X_n|(\theta_0^l,x_{0}),\dots,(\theta_{n-1}^l,x_{n-1})$ from $K_{\theta_{n-1}^l,l}(x_{n-1},\cdot)$.} \bigskip
\State{Update:
$$
\theta_n^l = \theta_{n-1}^l + \gamma_n H_l(\theta_{n-1}^l,X_n).
$$
Set $n=n+1$ and go to the start of 2..}
\end{algorithmic}
\end{algorithm}

\subsection{Debiasing Markovian Stochastic Approximation}

Let $\theta_{\star}^l$ be the solution of $h_l(\theta)$ and $\theta_{\star}$ be the solution of $h(\theta)=0$ and suppose that $\lim_{l\rightarrow\infty}\theta_{\star}^l=\theta_{\star}$.
We shall assume that these solutions are unique, although this is mainly for simplicity and we need only assume that there exist a collection of solutions.
Let $\mathbb{P}_L(l)$ be a positive probability on $\mathbb{N}_0$ then, we know from \cite{rhee} if one samples $l$ from $\mathbb{P}_L$ and then computes, independently
of $L$, $\theta^l_{\star}-\theta^{l-1}_{\star}$ (with $\theta^{-1}_{\star}:=0$), we have that
$$
\widehat{\theta}_{\star} = \frac{\theta^l_{\star}-\theta^{l-1}_{\star}}{\mathbb{P}_L(l)}
$$
is an unbiased estimator of $\theta_{\star}$.  The estimator returned by the MSA procedure after $N$ steps, $\theta_N^l$ is not equal in expectation to $\theta_{\star}^l$, which would be a requirement for unbiasedness (see e.g.~\cite{matti}), 
but under conditions, one would have
$$
\lim_{N\rightarrow\infty}\mathbb{E}[\theta_N^l] = \theta_{\star}^l.
$$
This latter result, suggests the following double randomization scheme used in \cite{ub_pf} and is now detailed in the sequel.

We begin by assuming that one can find a coupling, $\check{K}_{\theta,\theta',l,l-1}$, of $(K_{\theta,l},K_{\theta',l-1})$ for any $(\theta,\theta')\in\Theta^2$ and $l\in\mathbb{N}$; examples in the context of partially observed diffusions and Bayesian inverse problems can be found in \cite{ub_grad,disc_model} and we describe a particular scheme in Section \ref{sec:diffusions}. 
Let $\check{\nu}_{\theta,\theta',l,l-1}$ be any coupling of $(\nu_{\theta}^l,\nu_{\theta'}^{l-1})$.
Then let $\{N_p\}_{p\in\mathbb{N}_0}$ be
any sequence of increasing natural numbers, converging to infinity. Let $\mathbb{P}_P$ be any positive probability on $\mathbb{N}_0$. 
Then Algorithm \ref{alg:USMA} will, under assumptions, give unbiased estimates of $\theta_{\star}$. Algorithm \ref{alg:USMA} can be run $M-$times in parallel and averaged to reduce the variance of the estimator, if it exists. We remark that Algorithm \ref{alg:USMA} is completely generic, in that at this stage we have not explicitly specified any model $\pi_{\theta}$ and $\pi_{\theta}^l$, however, of course the details on $K_{\theta,l}$ and $\check{K}_{\theta,\theta',l,l-1}$ are rather important and this needs to be explained in a particular context.

\begin{algorithm}
\caption{Unbiased Markovian Stochastic Approximation (UMSA)}
\label{alg:USMA}
\begin{algorithmic}[1]
\item{Sample $l$ from $\mathbb{P}_L$ and $p$ from $\mathbb{P}_P$.}
\item If $l=0$ perform the following:
\begin{itemize}
\item{Set $\theta_0^l\in\Theta$, $n=1$ and generate $X_0\sim\nu_{\theta_0}^l$.}
\item{Sample $X_n|(\theta_0^l,x_{0}),\dots,(\theta_{n-1}^l,x_{n-1})$ from $K_{\theta_{n-1}^l,l}(x_{n-1},\cdot)$.}
\item{Update:
$$
\theta_n^l = \theta_{n-1}^l - \gamma_n H_l(\theta_{n-1}^l,X_n).
$$
If $n=N_p$ go to the next bullet point, otherwise go back to the second bullet point.}
\item{If $p=0$ return
$$
\widehat{\theta}_{\star} = \frac{\theta_{N_p}^l}{\mathbb{P}_p(p)\mathbb{P}_L(l)},
$$
otherwise  return
$$
\widehat{\theta}_{\star} = \frac{\theta_{N_p}^l-\theta_{N_{p-1}}^l}{\mathbb{P}_p(p)\mathbb{P}_L(l)}.
$$
}
\end{itemize}

\item Otherwise perform the following:
\begin{itemize}
\item Set $\theta_0^l=\theta_0^{l-1}\in\Theta$, $n=1$ and generate $(X_0^l,X_0^{l-1})\sim\check{\nu}_{\theta_0^l,\theta_0^{l-1},l,l-1}^l$.
\item Sample $X_n^l,X_{n}^{l-1}\Big|(\theta_0^l,\theta_0^{l-1},x_{0}^l,x_{0}^{l-1}),\dots,(\theta_{n-1}^l,\theta_{n-1}^{l-1},x_{n-1}^l,x_{n-1}^{l-1})$ from
 
\smallskip
$\check{K}_{\theta_{n-1}^l,\theta_{n-1}^{l-1},l,l-1}\left((x_{n-1}^l,x_{n-1}^{l-1}),\cdot,\right)$.

\item Update:
\begin{eqnarray*}
\theta_n^l & = & \theta_{n-1}^l + \gamma_n H_l(\theta_{n-1}^l,X_n^l), \\
\theta_n^{l-1} & = & \theta_{n-1}^{l-1} + \gamma_n H_{l-1}(\theta_{n-1}^{l-1},X_n^{l-1}).
\end{eqnarray*}
If $n=N_p$ go to the next bullet point, otherwise go back to the second bullet point.
\item If $p=0$ return
$$
\widehat{\theta}_{\star} = \frac{\theta_{N_p}^l-\theta_{N_p}^{l-1}}{\mathbb{P}_p(p)\mathbb{P}_L(l)},
$$
otherwise  return
$$
\widehat{\theta}_{\star} = \frac{\theta_{N_p}^l-\theta_{N_{p}}^{l-1}-\{\theta_{N_{p-1}}^l-\theta_{N_{p-1}}^{l-1}\}}{\mathbb{P}_p(p)\mathbb{P}_L(l)}.
$$
\end{itemize}
\end{algorithmic}
\end{algorithm}

\section{Application to Partially Observed Diffusions}\label{sec:diffusions}

\subsection{Model}

Consider the diffusion process on the filtered probability space $(\Omega,\mathscr{F},\{\mathscr{F}_t\}_{t\geq 0},\mathbb{P}_{\theta})$, $\theta\in\Theta\subseteq\mathbb{R}^{d_{\theta}}$
\begin{equation}
dX_t = a_{\theta}(X_t)dt + \sigma(X_t)dW_t\quad\quad X_0\sim\mu_{\theta}\label{eq:diff}
\end{equation}
where $X_t\in\mathbb{R}^d$, $\mu_{\theta}$ a probability measure on $(\mathbb{R}^d,\mathcal{B}(\mathbb{R}^d))$ with Lebesgue density denoted $\mu_{\theta}$ also, $a:\Theta\times\mathbb{R}^d \rightarrow\mathbb{R}^d$, $\{W_t\}_{t\geq 0}$ is a standard $d-$dimensional Brownian motion
and $\sigma:\mathbb{R}^d \rightarrow\mathbb{R}^{d\times d}$.  Denote $a_{\theta}^j$ as the $j^{th}-$component of $a_{\theta}$, $j\in\{1,\dots,d\}$
and $\sigma^{j,k}$ as the $(j,k)^{th}-$component of $\sigma$, $(j,k)\in\{1,\dots,d\}^2$.
We assume the following, referred to as (D1) from herein:
\begin{quote} 
The coefficients, for any fixed $\theta\in\Theta$, $a_{\theta}^j\in \mathcal{C}^2(\mathbb{R}^d)$ (twice continuously differentiable real-valued functions on $\mathbb{R}^d$) and $\sigma^{j,k} \in \mathcal{C}^2(\mathbb{R}^d)$, for $(j,k)\in\{1,\ldots, d\}$. 
Also for any fixed $x\in\mathbb{R}^d$, $a_{\theta}^j(x)\in\mathcal{C}(\Theta)$, $j\in\{1,\dots,d\}$.
In addition, $a_{\theta}$ and $\sigma$ satisfy 
\begin{itemize}
\item[(i)] {\bf uniform ellipticity}: $\Sigma(x):=\sigma(x)\sigma(x)^*$ is uniformly positive definite 
over $x\in \mathbb{R}^d$;
\item[(ii)] {\bf globally Lipschitz}:
for any fixed $\theta\in\Theta$, there exist a $C<+\infty$ such that 
$|a_{\theta}^j(x)-a_{\theta}^j(x')|+|\sigma^{j,k}(x)-\sigma^{j,k}(x')| \leq C \|x-x'\|_2$ ($\|\cdot\|_2$ is the $L_2-$norm) 
for all $(x,x') \in \mathbb{R}^d\times\mathbb{R}^d$, $(j,k)\in\{1,\dots,d\}^2$. 
\end{itemize}
\end{quote}
The assumption (D1) is not referred to again and strictly is not a requirement to apply any of the methodology that is to be discussed. However, some of the algorithms (e.g.~Algorithm \ref{alg:ql_def}) to be described, do not work well without this assumption, which is why it is stated.

By using the Girsanov theorem, for any $\varphi:\Theta\times\mathbb{R}^{nd}\rightarrow\mathbb{R}$, $\varphi_{\theta}(x_{t_1},\dots,x_{t_n})$, for $0\leq t_1<\cdots<t_n=T>0$ and $\mathbb{P}_{\theta}-$integrable, 
$$
\mathbb{E}_{\theta}[\varphi_{\theta}(X_{t_1},\dots,X_{t_n})] = \mathbb{E}_{\mathbb{Q}\otimes\textrm{Leb}}\left[\mu_{\theta}(X_0)\varphi_{\theta}(X_{t_1},\dots,X_{t_n})\frac{d\mathbb{P}_{\theta}}{d\mathbb{Q}}\right]
$$
where $\mathbb{E}_{\theta}$ denotes expectations w.r.t.~$\mathbb{P}_{\theta}$,
$$
\frac{d\mathbb{P}_{\theta}}{d\mathbb{Q}} = \exp\Big\{-\frac{1}{2}\int_{0}^T \|b_{\theta}(X_s)\|_2^2ds + \int_{0}^{T}b_{\theta}(X_s)^*dW_s\Big\}
$$
$\mathbb{Q}\otimes\textrm{Leb}$ is the product measure of a probability $\mathbb{Q}$ on the path $\{X_t\}_{t>0}$ and the Lebesgue measure on $X_0$,
$b_\theta(x) = \Sigma(x)^{-1}\sigma(x)^*a_{\theta}(x)$ is a $d-$vector and under $\mathbb{Q}$, $\{X_t\}_{t>0}$ solves $dX_t = \sigma(X_t)dW_t$ where 
$\{W_t\}_{t\geq 0}$ is a standard $d-$dimensional Brownian motion under $\mathbb{Q}$ also.  It is assumed that $a_{\theta},\varphi_{\theta}$ and $\sigma$ are such that $\mu_{\theta}(X_0)\varphi_{\theta}(X_{t_1},\dots,X_{t_n})\frac{d\mathbb{P}_{\theta}}{d\mathbb{Q}}$ is $\mathbb{Q}-$integrable for each fixed $T\geq 0$.
As it will be useful below, we write
$$
\frac{d\mathbb{P}_{\theta}}{d\mathbb{Q}} = \exp\Big\{-\frac{1}{2}\int_{0}^T \|b_{\theta}(X_s)\|_2^2ds + \int_{0}^{T}b_{\theta}(X_s)^*
\Sigma(X_s)^{-1}\sigma(X_s)^* dX_s\Big\}.
$$
Now if $\mu_{\theta}\varphi_\theta$ is differentiable w.r.t.~$\theta$, one has, under very minor regularity conditions that
$$
\nabla_{\theta}\log\left\{\mathbb{E}_{\theta}[\varphi_{\theta}(X_{t_1},\dots,X_{t_n})]\right\} = \mathbb{E}_{\overline{\mathbb{P}}_{\theta}}\left[\nabla_{\theta}\log\left(\mu_{\theta}(X_0)
\varphi_{\theta}(X_{t_1},\dots,X_{t_n})\frac{d\mathbb{P}_{\theta}}{d\mathbb{Q}}
\right)\right]
$$
where $\overline{\mathbb{P}}_{\theta}=\varphi_{\theta}\mathbb{P}_{\theta}/\mathbb{P}_{\theta}(\varphi_{\theta})$ and $\mathbb{P}_{\theta}(\varphi_{\theta})=\mathbb{E}_{\theta}[\varphi_{\theta}(X_{t_1},\dots,X_{t_n})]$. In this notation, the previously mentioned $\pi_{\theta}$ is the probability measure $\overline{\mathbb{P}}_{\theta}$, with $\varphi_{\theta}$ to be determined below.

Consider a sequence of random variables $(Y_{1},\dots,Y_{T})$, where $Y_{p}\in\mathbb{R}^{d_y}$, that are assumed to have joint Lebesgue density ($T\in\mathbb{N}$ is assumed from herein)
$$
p_{\theta}(y_{1},\dots,y_{T}|\{x_s\}_{0\leq s\leq T}) = \prod_{k=1}^T g_{\theta}(x_{k},y_{k})
$$
where, $g:\Theta\times\mathbb{R}^d\times\mathbb{R}^{d_y}\rightarrow\mathbb{R}^+$, for any $(\theta,x)\in\Theta\times\mathbb{R}^d$, $\int_{\mathbb{R}^{d_y}}g_\theta(x,y)dy=1$ and $dy$ is the Lebesgue measure. If one considers realizations of the random variables, $(Y_{1},\dots,Y_{T})$, then we have the state-space model with marginal likelihood:
$$
p_{\theta}(y_{1},\dots,y_{T}) = \mathbb{E}_{\theta}\left[\prod_{k=1}^T g_{\theta}(X_{k},y_{k})\right].
$$
Now, following the construction that is developed above, one would have the function $\varphi_{\theta}(x_{1},\dots,x_{T})= \prod_{k=1}^T g_{\theta}(x_{k},y_{k})$ and the expression:
$$
H(\theta,\{X_s\}_{s\in[0,T]}) = \nabla_{\theta}\log\Big(\mu_{\theta}(x_0)
\varphi_{\theta}(X_{1},\dots,X_{T})\frac{d\mathbb{P}_{\theta}}{d\mathbb{Q}}\Big) = 
$$
$$
\nabla_{\theta} \log\{\mu_{\theta}(X_0)\} + 
\sum_{k=1}^n \nabla_{\theta} \log\{g_{\theta}(X_{k},y_{k})\}
-\frac{1}{2}\int_{0}^T \nabla_{\theta}\|b_{\theta}(X_s)\|_2^2ds + \int_{0}^T \nabla_{\theta}b_{\theta}(X_s)^*\Sigma(X_s)^{-1}\sigma(X_s)^* dX_s.
$$

\subsubsection{Time Discretization}

Let $l\in\mathbb{N}_0$ be given and consider an Euler discretization of step-size $\Delta_l=2^{-l}$, $k\in\{1,2,\dots,\Delta_{l}^{-1}T\}$:
\begin{eqnarray}\label{eq:discretization}
\widetilde{X}_{k\Delta_l} & = & \widetilde{X}_{(k-1)\Delta_l} + a_{\theta}(\widetilde{X}_{(k-1)\Delta_l})\Delta_l + \sigma(\widetilde{X}_{(k-1)\Delta_l})[W_{k\Delta_l}-W_{(k-1)\Delta_l}].\label{eq:disc_state}
\end{eqnarray}
One can also define the function $H_l:\Theta\times(\mathbb{R}^d)^{\Delta_l^{-1}T+1}\rightarrow\mathbb{R}^{d_{\theta}}$
\begin{equation}\label{eq:h_l_def}
H_{l}(\theta,\widetilde{x}_{0},\dots,\widetilde{x}_{T}) = 
 \sum_{k=0}^{\Delta_l^{-1}T-1} \Big\{-\frac{\Delta_l}{2}\nabla_{\theta}\|b_{\theta}(\widetilde{x}_{k\Delta_l})\|_2^2
+\nabla_{\theta}b_{\theta}(\widetilde{x}_{k\Delta_l})^*\Sigma(\widetilde{x}_{k\Delta_l})^{-1}\sigma(\widetilde{x}_{k\Delta_l})^*[\widetilde{x}_{(k+1)\Delta_l}-\widetilde{x}_{k\Delta_l}]\Big\} +
\end{equation}
$$
\sum_{k=1}^T \nabla_{\theta} \log\{g_{\theta}(\widetilde{x}_{k},y_{k})\}  + \nabla_{\theta} \log\{\mu_{\theta}(\widetilde{x}_0)\}. 
$$
Then we have
$$
h_l(\theta) = \nabla_{\theta} \log(p_{\theta}^l(y_{1},\dots,y_{T})) := \frac{\mathbb{E}_{\theta}[\varphi_{\theta}(\widetilde{X}_{1},\dots,\widetilde{X}_{T})H_{l}(\theta,\widetilde{X}_{0},\dots,\widetilde{X}_{T})]}
{\mathbb{E}_{\theta}[\varphi_{\theta}(\widetilde{X}_{t_1},\dots,\widetilde{X}_{t_n})]}.
$$
The probability measure $\pi_{\theta}^l$ is simply the smoother of path $\widetilde{x}_{0},\dots,\widetilde{x}_{T}$ given the observations $y_1,\dots,y_T$. To be more precise, set $U_0=X_0$ and denote by $U_k$ the discretized path $\widetilde{X}_{k-1+\Delta_l},\dots,\widetilde{X}_{k}$, $k\in\{1,\dots,T\}$ and denote the transition kernel of $U_k$ given $U_{k-1}$ (as induced by \eqref{eq:disc_state})
as $M_{\theta,l}$ then one has that:
\begin{equation}\label{eq:pi_l_def}
\pi_{\theta}^l\left(d(u_0,\dots,u_T)\right) = \frac{\prod_{k=1}^Tg_{\theta}(\tilde{x}_k,y_{k})\mu_{\theta}(u_0)du_0\prod_{k=1}^T M_{\theta,l}(u_{k-1},du_k)}{\int_{(\mathbb{R}^d)^{\Delta_l^{-1}T+1}} \prod_{k=1}^Tg_{\theta}(\tilde{x}_k,y_{k})\mu_{\theta}(u_0)du_0\prod_{k=1}^T M_{\theta,l}(u_{k-1},du_k)}
\end{equation}
where we have omitted dependence on the data in $\pi_{\theta}^l$.

\subsection{Conditional Particle Filter}

To construct our algorithm for partially observed diffusions, we begin by describing the conditional particle filter (see \cite{andrieu}) which is a Markov kernel of invariant measure $\pi_{\theta}^l$. The simulation of the kernel is described in Algorithm \ref{alg:cond_pf_0}.

\begin{algorithm}[h]
\begin{enumerate}
\item{Input $U_{0}',\dots,U_T'$. Set $k=1$, sample $U_0^i$ independently from $\mu_{\theta}$, $a_{0}^i=i$ for $i\in\{1,\dots,N-1\}$.}
\item{Sampling: for $i\in\{1,\dots,N-1\}$ sample $U_{k}^i|U_{k-1}^{a_{k-1}^i}$ using the Markov kernel $M_{\theta,l}$. Set $U_{k}^N=U_k'$ and for $i\in\{1,\dots,N-1\}$, $(U_{0}^i,\dots,U_k^i)=(U_0^{a_{k-1}^i},\dots,U_{k-1}^{a_{k-1}^i},U_{k}^i)$. If $k=T$ go to 4..}
\item{Resampling: Construct the probability mass function on $\{1,\dots,N\}$:
$$
r_1^i = \frac{g_{\theta}(\widetilde{x}_k^i,y_k)}{\sum_{j=1}^Ng_{\theta}(\widetilde{x}_k^j,y_k)}.
$$
For $i\in\{1,\dots,N-1\}$ sample $a_k^i$ from $r_1^i$. Set $k=k+1$ and return to the start of 2..}
\item{Construct the probability mass function on $\{1,\dots,N\}$:
$$
r_1^i = \frac{g_{\theta}(\widetilde{x}_T^i,y_T)}{\sum_{j=1}^Ng_{\theta}(\widetilde{x}_T^j,y_T)}.
$$
Sample $i\in\{1,\dots,N\}$ using this mass function and return $(U_0^i,\dots,U_T^i)$.}
\end{enumerate}
\caption{Conditional Particle Filter at level $l\in\mathbb{N}_0$.}
\label{alg:cond_pf_0}
\end{algorithm}

\subsection{Coupled Conditional Particle Filter}

To describe the coupled conditional particle filter (CCPF), which is essentially described in \cite{ub_grad} and is a conditional particle filter associated to the algorithm in \cite{mlpf}, we require several objects that we shall now detail.

We begin with simulating the maximal coupling of two probability mass functions on $\{1,\dots,N\}$ in Algorithm \ref{alg:max_coup}. This will be needed in the resampling operation of the CCPF. Next we describe a coupling of $M_{\theta,l}$ and $M_{\theta',l-1}$, which we denote by $\check{M}_{\theta,\theta',l,l-1}$. This is given in Algorithm \ref{alg:ql_def}. It will be needed in the sampling step of the CCPF and is the well-known `synchronous coupling' of Euler discretizations. In Algorithm \ref{alg:ql_def}, $\mathcal{N}_d(0,\Delta_l I_d)$ is the $d-$dimensional Gaussian, 0 mean, $\Delta_l I_d$ covariance, $I_d$ is the $d\times d$ identity matrix.
Given Algorithm \ref{alg:max_coup} and Algorithm \ref{alg:ql_def} we are now in a position to describe the simulation of one step of the CCPF kernel $K_{\theta,\theta',l,l-1}$. The kernel will take as its input two trajectories at levels $l$ and $l-1$, denote
them as $U_{0:T}^l\in\mathbb{R}^{T\Delta_l^{l-1}d+1}$ and $U_{0:T}^{l-1}\in\mathbb{R}^{T\Delta_{l-1}^{l-1}d+1}$ and produce two new such trajectories.

\begin{algorithm}[h]
\begin{enumerate}
\item{Input: Two probability mass functions (PMFs) $(r_1^1,\dots,r_1^N)$ and $(r_2^1,\dots,r_2^N)$ on $\{1,\dots,N\}$.}
\item{Generate $U\sim\mathcal{U}_{[0,1]}$ (uniform distribution on $[0,1]$).}
\item{If $U<\sum_{i=1}^N \min\{r_1^i,r_2^i\}=:\bar{r}$ then generate $i\in\{1,\dots,N\}$ according to the probability mass function:
$$
r_3^i = \frac{1}{\bar{r}} \min\{r_1^i,r_2^i\}
$$
and set $j=i$.}
\item{Otherwise generate $i\in\{1,\dots,N\}$ and $j\in\{1,\dots,N\}$ independently according to the probability mass functions 
$$
r_4^i = \frac{1}{1-\bar{r}} (r_1^i - \min\{r_1^i,r_2^i\})
$$
and
$$
r_5^j = \frac{1}{1-\bar{r}} (r_2^j - \min\{r_1^j,r_2^j\})
$$
respectively.
}
\item{Output: $(i,j)\in\{1,\dots,N\}^2$. $i$, marginally has PMF $r_1^i$ and $j$, marginally has PMF $r_2^j$.}
\end{enumerate}
\caption{Simulating a Maximal Coupling of Two Probability Mass Functions on $\{1,\dots,N\}$.}
\label{alg:max_coup}
\end{algorithm}

\begin{algorithm}[h]
\begin{enumerate}
\item{Input $(x_0^l,x_0^{l-1})\in\mathbb{R}^{2d}$ and the level $l\in\mathbb{N}_0$.}
\item{Generate $V_{k\Delta_l}\stackrel{\textrm{i.i.d.}}{\sim}\mathcal{N}_d(0,\Delta_l I_d)$, for $k\in\{1,\dots,\Delta_l^{-1}\}$.}
\item{Run the recursion, for $k\in\{1,\dots,\Delta_l^{-1}\}$:
\begin{eqnarray*}
X_{k\Delta_l}^l & = & X_{(k-1)\Delta_l}^l + a_{\theta}(X_{(k-1)\Delta_l}^l)\Delta_l + \sigma(X_{(k-1)\Delta_l}^l)V_{k\Delta_l}
\end{eqnarray*}
}
\item{Run the recursion, for $k\in\{1,\dots,\Delta_{l-1}^{-1}\}$:
\begin{eqnarray*}
X_{k\Delta_l}^{l-1} & = & X_{(k-1)\Delta_l}^{l-1} + a_{\theta'}(X_{(k-1)\Delta_l}^{l-1})\Delta_l + \sigma(X_{(k-1)\Delta_l}^{l-1})[V_{(2k-1)\Delta_l}+V_{2k\Delta_l}] 
\end{eqnarray*}
}
\item{Return $\left((x_{\Delta_l}^l,\dots,x_1^l),(x_{\Delta_l}^{l-1},\dots,x_1^{l-1})\right)\in\mathbb{R}^{\Delta_l^{-1}d}\times\mathbb{R}^{\Delta_{l-1}^{-1}d}$.}
\end{enumerate}
\caption{Simulating the Kernel $\check{M}_{\theta,\theta',l,l-1}$.}
\label{alg:ql_def}
\end{algorithm}

\begin{algorithm}[h]
\begin{enumerate}
\item{Input $(U_{0:T}^{',l},U_{0:T}^{',l-1})$. Set $k=1$, sample $U_0^{i,l}$ independently from $\mu_{\theta}$ and independently $U_0^{i,l-1}$ from $\mu_{\theta'}$, $a_{0}^{i,l}=a_0^{i,l-1}=i$ for $i\in\{1,\dots,N-1\}$.}
\item{Sampling: for $i\in\{1,\dots,N-1\}$ sample $(U_{k}^{i,l},U_k^{i,l-1})|(U_{k-1}^{a_{k-1}^{i,l},l},U_{k-1}^{a_{k-1}^{i,l-1},l-1})$ using the Markov kernel $\check{M}_{\theta,\theta',l,l-1}$ in Algorithm \ref{alg:ql_def}. Set $U_{k}^{N,l}=U_k^{',l}$,  $U_{k}^{N,l-1}=U_k^{',l-1}$ and for $(i,s)\in\{1,\dots,N-1\}\times\{l-1,l\}$, $U_{0:k}^{i,s}=(U_0^{a_{k-1}^{i,s},s},\dots,U_{k-1}^{a_{k-1}^{i,s},s},U_{k}^{i,s})$. If $k=T$ go to 4..}
\item{Resampling: Construct the probability mass functions on $\{1,\dots,N\}$:
$$
r_1^i = \frac{g_{\theta}(\widetilde{x}_k^{i,l},y_k)}{\sum_{j=1}^Ng_{\theta}(\widetilde{x}_k^{j,l},y_k)} \quad\textrm{and}\quad
r_2^i = \frac{g_{\theta'}(\widetilde{x}_k^{i,l-1},y_k)}{\sum_{j=1}^Ng_{\theta'}(\widetilde{x}_k^{j,l-1},y_k)}.
$$
For $i\in\{1,\dots,N-1\}$ sample $a_k^{i,l}$ and $a_k^{i,l-1}$ using the maximal coupling (Algorithm \ref{alg:max_coup}) with PMFs $r_1^i$ and $r_2^i$. Set $k=k+1$ and return to the start of 2..}
\item{Construct the probability mass functions on $\{1,\dots,N\}$:
$$
r_1^i = \frac{g_{\theta}(\widetilde{x}_T^{i,l},y_T)}{\sum_{j=1}^Ng_{\theta}(\widetilde{x}_T^{j,l},y_T)} \quad\textrm{and}\quad
r_2^i = \frac{g_{\theta'}(\widetilde{x}_T^{i,l-1},y_T)}{\sum_{j=1}^Ng_{\theta'}(\widetilde{x}_T^{j,l-1},y_T)}
$$
Sample $(i,j)\in\{1,\dots,N\}^2$ using these mass functions via Algorithm \ref{alg:max_coup} and return $(U_{0:T}^{i,l},,U_{0:T}^{i,l-1})$.}
\end{enumerate}
\caption{Coupled Conditional Particle Filter at level $l\in\mathbb{N}$.}
\label{alg:ccpf}
\end{algorithm}

\subsection{Final Algorithm}\label{sec:final_algo}

The method that we use is then reasonably simple, given appropriate choices of $\mathbb{P}_L$ and $\mathbb{P}_p$, which is a topic to be discussed in the next section. 
\begin{itemize}
\item{Run Algorithm \ref{alg:USMA} with the choice of $\pi_{\theta}^l$ as in \eqref{eq:pi_l_def} and $H_l(\theta,U_{0:T})$ as in \eqref{eq:h_l_def}.}
\item{The kernel $K_{\theta,l}$ is sampled as in  Algorithm \ref{alg:cond_pf_0}.}
\item{The kernel $\check{K}_{\theta,\theta',l,l-1}$ is sampled as Algorithm \ref{alg:ccpf}.}
\end{itemize}

At this stage several remarks are important. Firstly, the estimator that we have employed is the so-called single term estimator for randomization (see \cite{rhee}). This can be improved by using independent sum estimators and we refer to \cite{ub_grad,disc_model,rhee,matti} for details. Secondly, in Algorithm \ref{alg:max_coup}, step 4.~can be improved by sampling from any coupling of $(r_4^i,r_5^j)$, although for simplicity, we do not do this. Thirdly, in Algorithm \ref{alg:ccpf}, we can improve step 1.~by sampling $(U_0^{i,l},U_0^{i,l-1})$ from a coupling of $(\mu_{\theta},\mu_{\theta'})$. Finally, for the probability measure $\nu_{\theta}^l$ (see e.g.~Algorithm \ref{alg:USMA}) we use the simulation of $U_0$ from $\mu_{\theta}$ and the rest of the path is generated using the recursion \eqref{eq:disc_state}. To sample the coupling $\check{\nu}^l_{\theta,\theta}$ we simply copy $U_0^{l}$ to obtain $U_0^{l-1}$ and generate the two trajectories all the way to time $T$ using $T$ applications of $\check{M}_{\theta,\theta',l,l-1}$ in Algorithm \ref{alg:ql_def}.

\subsection{Theoretical Results}\label{sec:theory}

We now give the main theoretical result of the paper which is that, under assumptions and modifications, the estimator that we have introduced is unbiased. By unbiased, we mean that the expected value
of the estimator is exactly $\theta_{\star}$. The estimator that we analyze, however, is slightly modified from the procedure that is discussed in Section \ref{sec:final_algo} as we use the method of reprojection (see e.g.~\cite{andr3} and the references therein) within the stochastic approximation method. This is essentially a minor modification which makes the mathematical analysis of stochastic approximation methods far easier. We remark, that in practice we never use the reprojection method and hence its description is relegated to the Appendix (Appendix \ref{app_sec:mod}).

The theoretical result is given under a collection of assumptions that are listed and discussed in the Appendix (Appendix \ref{app_sec:ass}) and are termed  \hyperref[assump:A1]{(A1-9)}. Below, we write $\mathbb{E}[\cdot]$ to denote the expectation w.r.t.~the randomness of our estimator that has been produced under the modification detailed in 
Appendix \ref{app_sec:mod}. We then have the following result.
 
 \begin{theorem}
 \label{thm:unbiased}
 Assume \hyperref[assump:A1]{(A1-9)}. Then we have that $\mathbb{E}[\widehat{\theta}_*] = \theta_{\star}$.
\end{theorem}

\begin{proof}
This follows by Theorem \ref{theo:you_idiot} in the Appendix, the bounded convergence theorem and several simple calculations that are omitted for brevity.
\end{proof}

The result that is given is essentially the minimum one would like to provide. One would also like to show that the estimator also has a finite variance and a finite expected cost (or cost that this finite `with high probability') as is the case for several other doubly randomized estimators \cite{ub_grad,disc_model,par_uq,ub_pf}. The challenge here is to show that the coupling of the estimators across consecutive levels of discretizations are sufficiently close as a function of $l$; typically this is measured through the second moment of the difference, and bounds are often of the type $\mathcal{O}(\Delta_l^{\beta})$ - see \cite{ub_grad,disc_model,par_uq,ub_pf}. In our context, this require a rather intricate analysis of the coupling of the conditional particle filter and its interplay with the stochastic approximation method. In the proofs given in the Appendix, all of the bounds that are proved, would explode with the discretization level and so, despite the technical sophistication of our proofs an even more intricate proof would be needed. As a result, in this paper we simply consider the empirical performance of our estimator and leave further mathematical analysis to future work.

\section{Numerical Simulations}\label{sec:numerics}
In this section, we test our algorithm on two models and compare it against the methodology proposed in \cite{ub_grad}.

\subsection{Ornstein-Uhlenbeck Process}

Consider the Ornstein-Uhlenbeck (OU) process $\{X_t\}_{t\geq 0}$ defined by
\begin{equation*}
\begin{split}
        dX_{t} = -\theta X_t dt + \sigma dW_t, \quad X_0=x_0\in\mathbb{R}, ~~ t\in [0,T],\\
\end{split}
\end{equation*}
where $\theta,\sigma\in\mathbb{R}_+$ and $T\in\mathbb{N}$. Let $\varsigma\in\mathbb{R}_+$, then the observations $\{Y_k\}_{k=1}^T$ are taken at unit times as $Y_k|X_k=x_k\sim\mathcal{N}(x_k, \varsigma^2)$, $k\in\{1,...,T\}$, generated from the true parameter $\theta^*=0.5$. We set $x_0=100$, $T=25$, $\sigma=0.4$ and $\varsigma=1$. 
We consider the OU process in this example because the likelihood can be computed analytically, which then can be used in the gradient descent method to obtain the maximum likelihood estimator (MLE). We apply our methodology presented in Algorithm \ref{alg:USMA} to estimate $\theta$ and compare our estimator to the MLE obtained from the exact model. Let $\tilde{\theta}$ be the MLE obtained from running the gradient descent. For each $M\in\{2^k~:~ 3\leq k\leq 13\}$, we run $M$ independent copies of our algorithm in parallel. Let $\hat{\theta}_1,...,\hat{\theta}_M$ be the estimates obtained from each run. For each $M$, define $\hat{\theta}_{M}^* = \frac{1}{M}\sum_{i=1}^M\hat{\theta}_i$. The MSE is then calculated as 
$$\frac{1}{100} \sum_{i=1}^{100}|\hat{\theta}_{M}^{*,i}-\tilde{\theta}|^2,$$ 
which was estimated by running the described procedure above 100 times. The number of particles used in the CPF or the CCPF is 50. We set $\mathbb{P}_L(l)=2^{-1.5l} \mathbb{I}_{\{l_{\text{min}},\cdots,l_{\text{max}}\}}(l)$ for some $l_{\text{min}},l_{\text{max}}\in \mathbb{N}\cup \{0\}$, where $\mathbb{I}_A$ denotes the indicator function on a set $A$. Given $l$ sampled from $\mathbb{P}_L$, we sample $p$ from $\mathbb{P}_{P|L}(p|l) \propto g(p|l)$ and set $N=N_0 ~2^p$, where
\begin{align*}
g(p|l) =  \left\{\begin{array}{lcl}
2^{5-p} & \text{if} & p \in \{p_{\text{min}},\cdots,5 \wedge (l_{\text{max}} - l)\},\\
2^{-p}~ p~[\log_2(p)]^2 &\text{if} & 5 < p \leq p_{\text{max}},
\end{array} \right.
\end{align*}
for some $p_{\text{min}}, p_{\text{max}} \in \mathbb{N} \cup \{0\}$. This choice is very similar to the one used in \cite{ub_pf}. We set $N_0=10$, $l_{\text{min}} = 3$, $l_{\text{max}}=12$, $p_{\text{min}}=1$ and $p_{\text{max}}=12$.
In Figure \ref{fig:ou_1}, we plot the MSE against both the CPU run time and $M$. The run time here is the sum of CPU run times of of all the $M$ processes that were run in parallel. Figure \ref{fig:ou_1} shows that the MSE scales as $M^{-1}$ which agrees with our theory that the estimator $\hat{\theta}_M^*$ is unbiased and has a variance that scales as $M^{-1}$.

\begin{figure}[h!]
    \centering
    \includegraphics[scale=0.5]{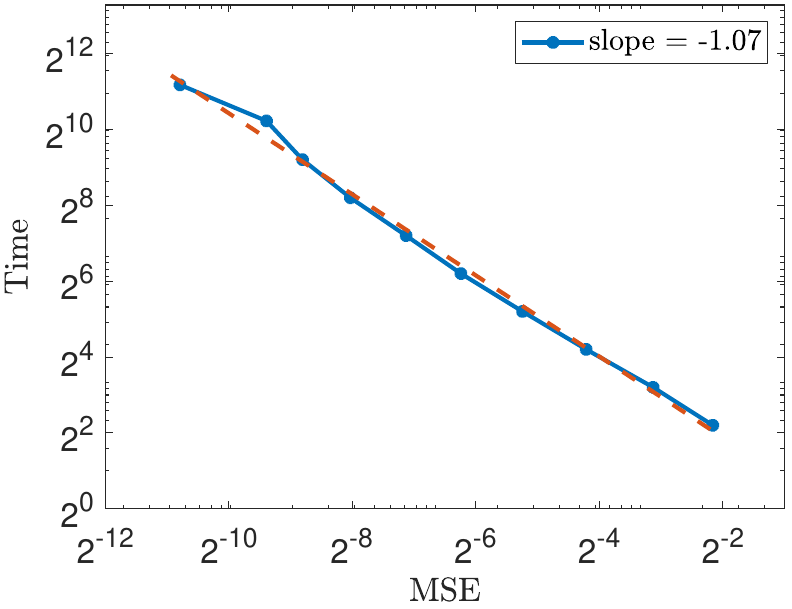}
    \includegraphics[scale=0.5]{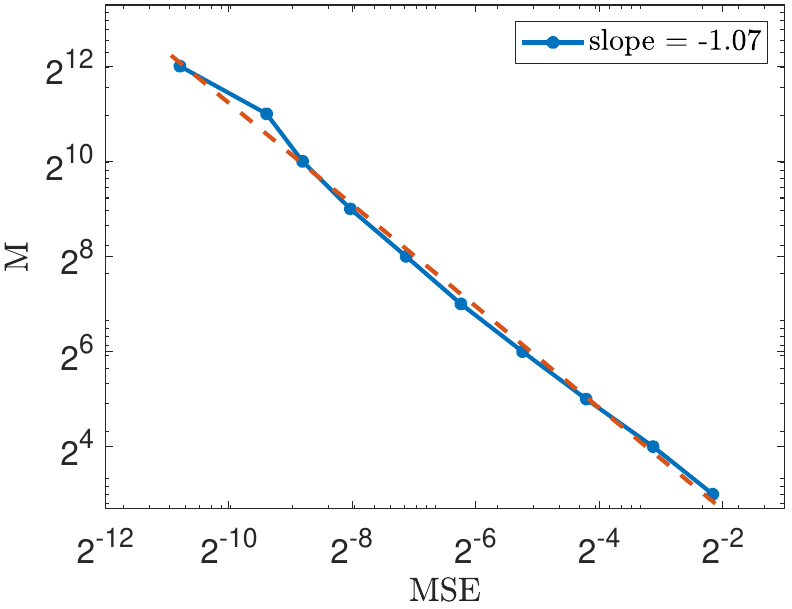}
    \caption{(OU Model) The UMSA algorithm is applied to estimate the drift parameter of the OU process. Left: MSE against run time. Right: MSE against $M$.}
    \label{fig:ou_1}
\end{figure}

Next, we compare our UMSA algorithm against the algorithm proposed in \cite{ub_grad}, where the latter gives an unbiased estimator for the score function, that is, the gradient of the log-likelihood, which is subsequently utilized in a stochastic gradient method to estimate the static parameters. Since usually it is difficult to decide when is the right time to stop the stochastic gradient method in the algorithm of \cite{ub_grad}, we only measure the time needed to compute an unbiased estimate of the score function in a neighborhood of $\theta=0.5$. Notice however that the run time of the algorithm in \cite{ub_grad} that is needed to estimate $\theta$ is in general much more (the average cost is proportional to the number of iterations used in the stochastic gradient method times the average time needed to unbiasedly estimate the score function). On the other hand, as for our algorithm, we compute the median time of running Algorithm \ref{alg:USMA} only once. We run 1000 simulations of each of the aforementioned procedures to generate box-plots for the run time. As we can see in Figure \ref{fig:ou_time_compar}, the median time needed to provide one estimate of $\theta$ using UMSA is less than that needed to generate an unbiased estimate of the score function in a neighborhood of $\theta=0.5$. This indicates that the run time needed to compute the estimator $\hat{\theta}_M^*$ is indeed going to be less than that needed to estimate $\theta$ using the method in \cite{ub_grad} since the $M$ simulations of UMSA algorithm are run in parallel, and therefore, the median run time to compute $\hat{\theta}_M^*$ is almost the same as the median run time of running Algorithm \ref{alg:USMA} only once.

\begin{figure}[h!]
    \centering
    \includegraphics[scale=0.55]{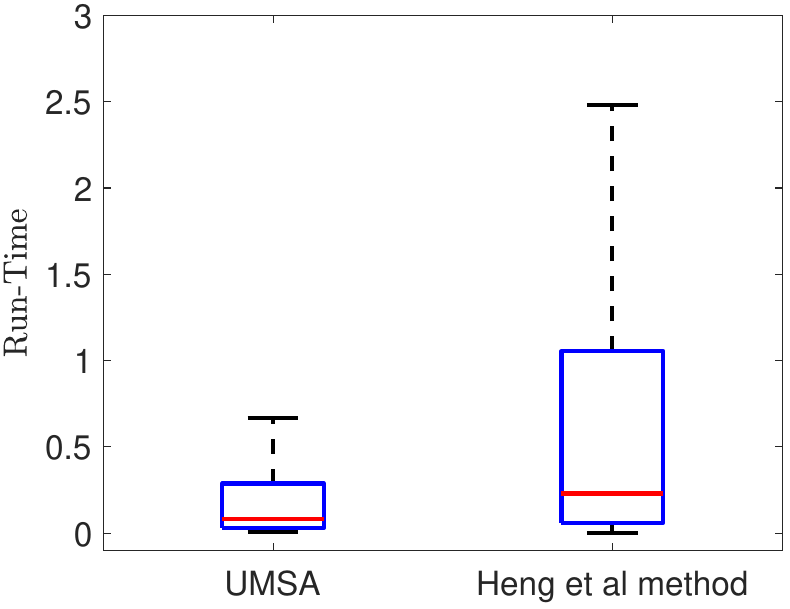}
    \caption{(OU Model) Comparison of the time needed to compute one unbiased estimate of $\theta$ using UMSA and the time needed to compute an unbiased estimate of the score function in a small neighborhood of $\theta^*=0.5$ using the method in \cite{ub_grad}. The box-plots are generated from 1000 runs of each procedure.}
    \label{fig:ou_time_compar}
\end{figure}

\subsection{Diffusion Model for Red Kangaroos}

In this example, we look at an application from population ecology to predict the dynamics of a population of red kangaroos (Macropus rufus) in New South Wales, Australia \cite{kangaroo}. The latent population size $Z=\{Z_t\}_{t\geq t_1}$ ($t_1>0$) is assumed to follow
a logistic diffusion process with environmental variance \cite{dennis,knape},  given as
\begin{equation}
\label{eq:kang}
dZ_t = (\theta^2_3/2 + \theta_1-\theta_2Z_t)Z_tdt + \theta_3Z_tdW_t, \qquad Z_{t_1} \sim \mathcal{L}{\mathcal{N}}(5,10^2),
\end{equation}
where $\mathcal{L}\mathcal{N}$ denotes the log-normal distribution. The parameters, $\theta_1 \in \mathbb{R}$ and $\theta_2>0$ can be thought of as coefficients that describe how the growth rate varies with population size. However, as $\theta_3$ appears in the diffusion coefficient of \eqref{eq:kang} we apply the Lamperti transformation $X_t =\log(Z_t)/\theta_3$. Applying It\^{o}'s formula, the new process $(X_t)_{t\geq 0}$ satisfies
\begin{equation}
    dX_t = a_{\theta}(X_t)dt + dW_t, \qquad X_{t_1}\sim \mathcal{N}\left(\frac{5}{\theta_3}, \frac{100}{\theta_3^2}\right),
\end{equation}
with $a_{\theta}(x) = \theta_1/\theta_3 - (\theta_2/\theta_3)\exp{(\theta_3 x)}$. The observations are double transect counts (pairs of non-negative integers) at irregular times $t_1,...,t_P$, with $P=41$, denoted by $Y_{t_1},\dots,Y_{t_P}$ and $Y_t \in \mathbb{R}^2$. It is assumed that the observations $\{Y_t\}_{t=t_1}^{t=t_P}$ are conditionally independent given $\{X_t\}_{t_1 \leq t \leq t_P}$ and are negative binomial distributed, precisely the density of $Y_t$ given $X_t$ is 
\begin{equation*}
    g_{\theta}(x_t,y_t) = \mathcal{NB}(y_t^1;\theta_4, \exp\{\theta_3 x_t\})~\mathcal{NB}(y_t^2;\theta_4, \exp\{\theta_3 x_t\}), \quad \theta_4>0, \quad t \in \{t_1,\cdots,t_P\},
\end{equation*}
where
$\mathcal{NB}(y;r,\mu) = \frac{\Gamma(y+r)}{\Gamma(r)y!}\left( \frac{r}{r+\mu} \right)^r \left(\frac{\mu}{r+\mu}\right)^y$, for $y\in\mathbb{N}\cup\{0\}$, $r\in(0,\infty)$, and $\mu\in(0,\infty)$. The goal is to estimate the parameters $\boldsymbol{\theta} = (\theta_1,\theta_2,\theta_3,\theta_4)\in \mathbb{R}\times (0,\infty)^3$. 

We employ $3$ as the minimum discretization level. For every $l\geq 3$ set the step size $\Delta_l=2^{-l}$ and replace the irregular times $t_1,...,t_P$ by rounding each of them to the closest number of the form $t_1+k\Delta_l$ for integers $k$, that is, we replace the irregular time $t_i$ with $\tilde{t}_i = \lfloor \frac{t_i-t_1}{\Delta_l}+\frac{1}{2}\rfloor + t_1$, $i\in \{1,\cdots,P\}$.
For each $M\in\{2^k~:~ 3\leq k\leq 13\}$, we run $M$ independent repeats of our algorithm in parallel. Let $\hat{\boldsymbol{\theta}}_1,...,\hat{\boldsymbol{\theta}}_M$ be the estimates obtained from running $M$ copies of our algorithm. For each $M$, define $\hat{\boldsymbol{\theta}}_{M}^* = \frac{1}{M}\sum_{i=1}^M\hat{\boldsymbol{\theta}}_i$. Figure \ref{fig:kang} is a log-log scale plot of the MSE versus the CPU run time of the algorithm. For each $M\in\{2^k~:~ 3\leq k\leq 12\}$, the MSE is calculated as 
$$
\textup{MSE}_j=\frac{1}{100} \sum_{i=1}^{100}|\hat{\boldsymbol{\theta}}_{M,j}^{*,i}-\hat{\boldsymbol{\theta}}_{2^{13},j}^*|^2, \quad \text{for } 1\leq j\leq 4,
$$ 
where $\hat{\boldsymbol{\theta}}_{k,j}$ denotes the $j$-th component of the the vector $\hat{\boldsymbol{\theta}}_{k}$. Here, the reference value $\hat{\boldsymbol{\theta}}_{2^{13},j}^*$ is computed by running Algorithm \ref{alg:USMA} $2^{13}$ times. The reason of considering $\hat{\boldsymbol{\theta}}_{2^{13}}$ as a proxy for the true MLE is that the likelihood for this model is intractable. We set $\mathbb{P}_L$, $\mathbb{P}_P$ and the number of particles used in CPF and CCPF similar to those in the previous example.

Figure \ref{fig:kang} shows that the run time scales approximately as MSE$^{-1}$. The fact that the estimator has the same rate as Monte-Carlo is a consequence of the unbiasedness property of our estimator as proven in Theorem \ref{thm:unbiased}. Next, we compare our UMSA algorithm against the algorithm proposed in \cite{ub_grad}, similar to what we did in the previous example. Figure \ref{fig:kang_time_compar} shows that the median run time needed to compute $\hat{\theta}_M^*$ is less than that needed to compute one unbiased estimate of the score function in a neighborhood of $\boldsymbol{\theta}^*$. Again, this example also shows that our method outperforms the methodology of \cite{ub_grad}.

\begin{figure}[h!]
    \centering
    \includegraphics[scale=0.5]{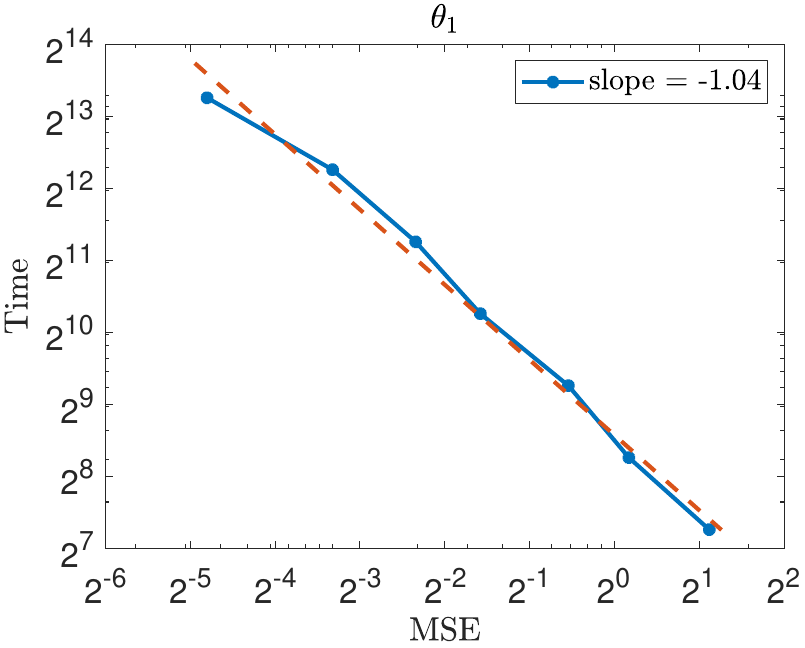}
    \includegraphics[scale=0.5]{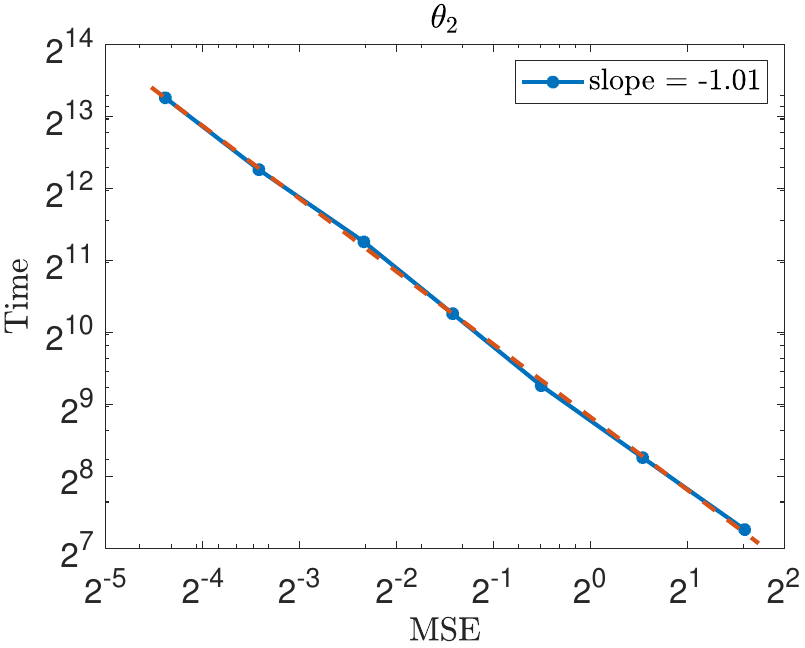}\\
    \includegraphics[scale=0.5]{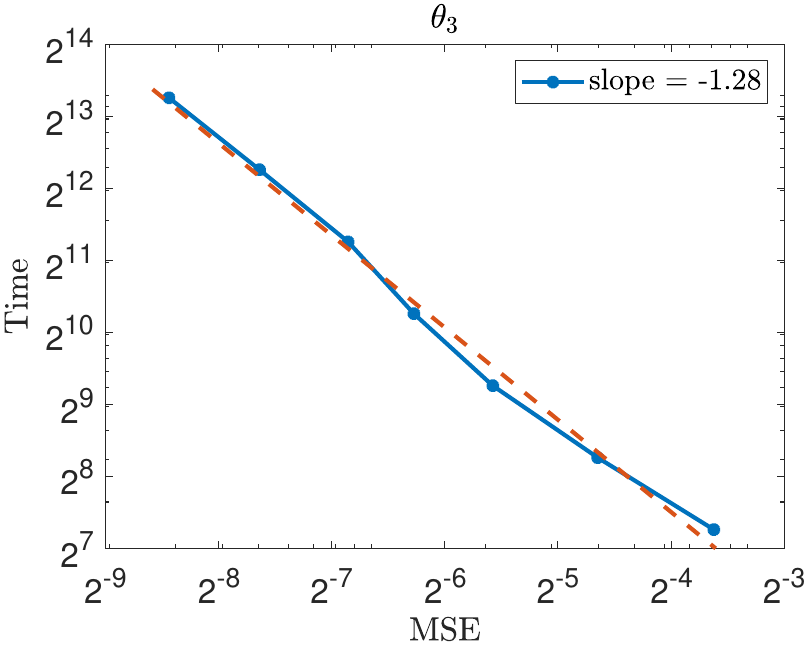}
    \includegraphics[scale=0.5]{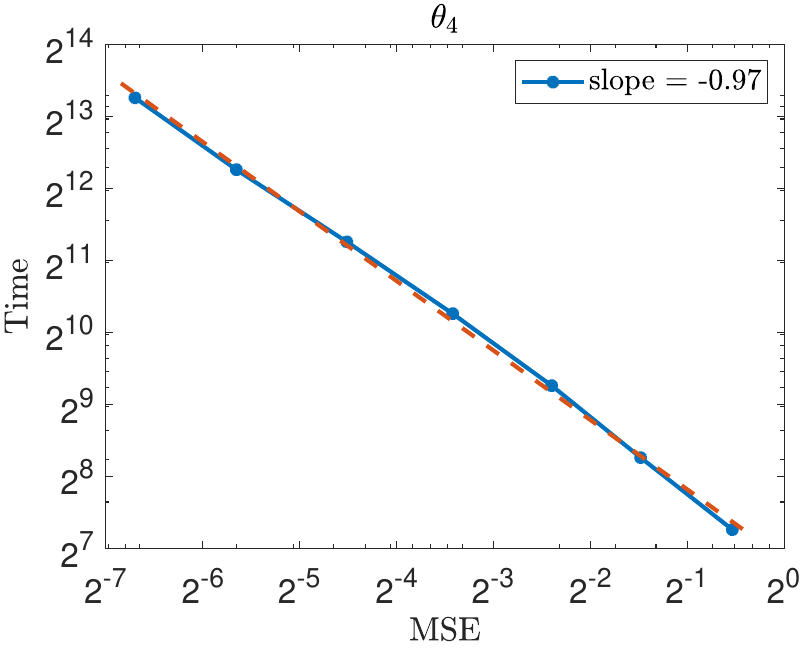}
    \caption{(Red Kangaroo Model) The UMSA algorithm applied to estimate the parameters of the red Kangaroo model. The plots show the MSE vs. run time for each parameter on a log-log scale.}
    \label{fig:kang}
\end{figure}

\begin{figure}[h!]
    \centering
    \includegraphics[scale=0.55]{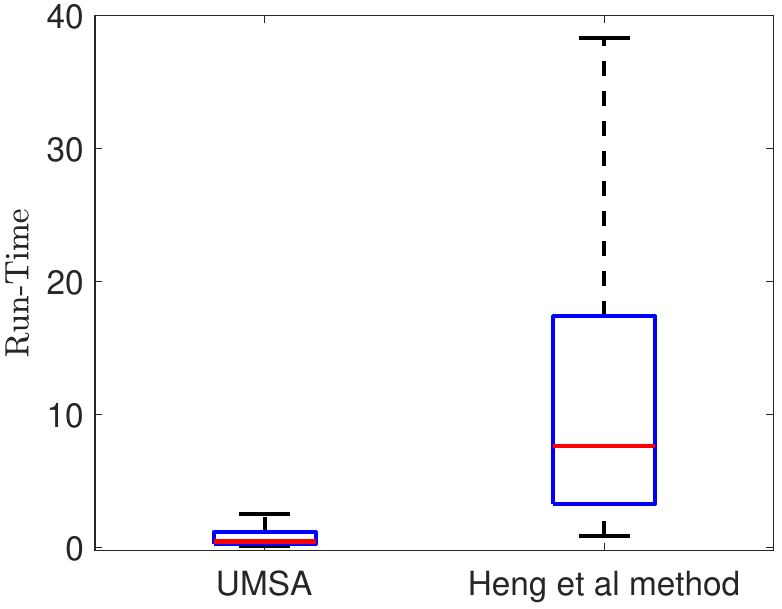}
    \caption{(Red Kangaroo Model) Comparison of the time needed to run the UMSA algorithm only once and the time needed to compute an unbiased estimate of the score function in a small neighborhood of $\boldsymbol{\theta}^*=(2.397,4.429\times10^{-03}, 0.84, 17.631)$ using the method in \cite{ub_grad}. The box-plots are generated from 1000 runs of each procedure.}
    \label{fig:kang_time_compar}
\end{figure}

\subsubsection*{Acknowledgements}

All authors were supported by KAUST baseline funding.

\appendix

\section{Proofs}\label{app:proofs}
\newcommand{\mathbbm}[1]{\text{\usefont{U}{bbm}{m}{n}#1}}
\newtheorem{lemma}{Lemma}
\newtheorem{remark}{Remark}
\newtheorem{corollary}{Corollary}

\subsection{Structure}

The appendix is written in a manner such that it should be read in order. The main idea is to verify the almost sure convergence of a slightly modified estimator, as will be described below. To that end, we will essentially seek to show that \cite[Theorem 5.5]{andr3} holds in our context, which verifies exactly the afore-mentioned property. This is stated in Theorem 
\ref{theo:you_idiot} later on, on which we give the proof in series of technical lemmata. The structure is then as follows. In Section \ref{app_sec:not} we give some additional notations. In Section \ref{app_sec:mod} we provide the slightly modified estimator. In Section \ref{app_sec:ass} we detail our assumptions and the main theorem. Finally in Section \ref{app_sec:proof} the proofs for the main theorem are presented.

\subsection{Notations}\label{app_sec:not}

Throughout the section we drop the $\tilde{X}$ and $U$ notations that are used in Section \ref{sec:diffusions} and simply use $X$. The notation $C$, $C_1$ etc are used to denote constants whose value may change from line-to-line but do not depend on $\theta$: dependencies on the model and simulation parameters will be clear from the context.

For every $\theta\in\Theta$ let $Q_{\theta,l}:\mathbb{R}^d\times\mathcal{B}(\mathbb{R}^d)\rightarrow [0,1]$ be the kernel of the discretized equation \eqref{eq:discretization}. That is
$$ Q_{\theta,l}(x,dz) = \frac{1}{\sqrt{(2\pi|\sigma(z)|^2)^d}}e^{(x-y-a_{\theta}(z)\Delta_l)^{\top}\sigma(z)\sigma(z)^{\top}(x-y-a_{\theta}(z)\Delta_l)}dz $$
where for a matrix $A$ the notation $A^{\top}$ denotes the transpose of $A$, $|A|$ denotes its determinant, and $dz$ is the Lebesgue measure on $\mathbb{R}^d$.
For every $\theta\in\Theta$ we define the unit-step kernel $M_{\theta,l}:(\mathbb{R}^d)^{\Delta_l^{-1}}\times\mathcal{B}((\mathbb{R}^d)^{\Delta_l^{-1}})\rightarrow [0,1]$ as
$$ M_{\theta,l}(x,dz) = Q_{\theta,l}(x_{\Delta_l^{-1}},dz_1)\prod_{i=2}^{\Delta_l^{-1}} Q_{\theta,l}(z_{i-1},dz_i).$$
Let $\{X_i\}_{i=0}^{T\Delta_l^{-1}}$ be the path of the process generated by $Q_{\theta,l}$ until time $T$. For every $1\leq k\leq\Delta_l^{-1}-1$ and a set $A\in\mathcal{B}((\mathbb{R}^d)^{\Delta_l^{-1}})$ of discrete paths the values of $M_{\theta,l}(x,A)$ is the probability of the event $\{X_i\}_{i=1+\Delta_l^{-1}k}^{\Delta^{-1}(k+1)}\in A$ given $\{X_i\}_{i=1+\Delta_l^{-1}(k-1)}^{\Delta^{-1}k} = x$. Without loss of generality and for ease of notation we extend the definition of $M_{\theta,l}$ to $\left((\mathbb{R}^d)^{\Delta_l^{-1}}\cup \mathbb{R}^d\right)\times\mathcal{B}((\mathbb{R}^d)^{\Delta_l^{-1}})$ since $M_{\theta,l}$ uses only the last component of $x$. This allows us to think of $M_{\theta,l}(x,A)$ as the transition probability of the event $\{\{X_i\}_{i=1}^{\Delta_l^{-1}}\in A\}$ given $X_0=x\in\mathbb{R}^d$.

Let  $y_1,...,y_T$ be the partial observations at times $1,...,T$. For every $l\in\mathbb{N}_0$ and $\theta\in\Theta$ define $G_{\theta,t}^l:\left(\mathbb{R}^{d}\right)^{\Delta_l^{-1}}\rightarrow\mathbb{R}$ by $ G_{\theta,t}^l(x)=g_{\theta}(x_{\Delta_l^{-1}},y_t)$ for $1\leq t\leq T$. Let $A_{1:T-1}^{1:N}$ be the resampling indices and $(F_t)_{t=0}^{T}$ be the indices of the path returned by conditional particle filter, these indices satisfy $A_{t}^N=N$ and $F_t = A_t^{F_{t+1}}$ for $1\leq t\leq T-1$ and $F_0=F_1$, hence deciding the value of $F_T$ and the resampling indices $A_{1:T-1}^{1:N-1}$ decides the indices of the whole path that conditional particle filter will return. The kernel of the conditional particle filter $K_{\theta,l}$ at level $l$ is defined by its action on bounded functions $\varphi:\left(\mathbb{R}^{d}\right)^{T\Delta^{-1}_l+1}\rightarrow \mathbb{R}$ given a discrete path $z\in\left(\mathbb{R}^{d}\right)^{T\Delta^{-1}_l+1}$ as 
\begin{equation}\label{eq:cpf_kernel}
    \begin{split}
        K_{\theta,l}(\varphi)(z)=\sum_{1\leq F_T,A_{1:T-1}^{1:N-1}\leq N}\int \delta_{z}(dx^{N}_{0:T})\frac{G_{\theta,T}^l(x^{F_T}_T)}{\sum_{k=1}^N G_{\theta,T}^l(x^k_T)}\prod_{j=1}^{N-1}\mu_{\theta}(dx_0^{j})M_{\theta,l}(x_0^j,dx^j_1)\\\times\prod_{t=2}^{T}\frac{G_{\theta,t-1}^l(x^{A^j_{t-1}}_{t-1})}{\sum_{k=1}^N G_{\theta,t-1}^l(x^k_{t-1})}M_{\theta,l}(x^{A_{t-1}^j}_{t-1},dx^j_t)\varphi(x_0^{F_0},...,x_T^{F_T}),
    \end{split}
\end{equation}
where the $x_t^j$ are elements in $\left(\mathbb{R}^{d}\right)^{\Delta_l^{-1}}$ for $1\leq t\leq T$, $1\leq j\leq N$ and $x_0^j\in\mathbb{R}^d$ for $1\leq j\leq N$, and the notation $\delta_z(dx_{0:T}^N) = \prod_{t=0}^T\delta_{z_t}(dx_t^N)$. 
The integral \eqref{eq:cpf_kernel} is finite and well-defined for any bounded function $\varphi$ since assumption \hyperref[assump:A5]{(A5)} guarantees that the functions $G^l_{\theta,t}$ is bounded away from zero for every $0\leq t\leq T$. It is well-known that for every $\theta$ the kernel $K_{\theta,l}$ is $\psi$-irreducible and aperiodic with $\pi^l_{\theta}$ as its stationary distribution \cite{andrieu2012}, therefore asymptotically we are able to sample from the filtering distribution $\pi^l_{\theta}$ by iteratively sampling from $K_{\theta,l}$.

\subsection{Modified Algorithm}\label{app_sec:mod}

We use the function $-H_l$ defined in \eqref{eq:h_l_def} to run the Markovian stochastic approximation algorithm in \cite{andr3} as follows. 
Let $\{\gamma_n\}_{n\in\mathbb{N}_0}$ be a sequence of positive real numbers such that $\sum_{n\in\mathbb{N}_0} \gamma_n = \infty$ and $\sum_{n\in\mathbb{N}_0} \gamma_n^{2}<\infty$. Suppose we have sequence of increasing compact sets $\{\Theta_n\}_{n\in\mathbb{N}_0}$ such that $\bigcup_n \Theta_n = \Theta$ and $\Theta_n\subset\textrm{int}(\Theta_{n+1})$. Let $\{\epsilon_n\}_{n\in\mathbb{N}}$ be a sequence of positive real numbers that converges to $0$. For every $l\in\mathbb{N}_0$ we define the stochastic approximation with re-projections defined in \cite[Section 3.3]{andr3} as a sequence of pairs $(\theta_n^l, X_n^l)\in\Theta\times\left(\mathbb{R}^{d}\right)^{T\Delta_l^{-1}+1}$ defined iteratively by
\begin{equation}\label{eq:msa}
\begin{split}
    &\textrm{ Sample } X^l_{n+1}\sim K_{\theta_n^l,l}(X_n^l, \cdot)\\
    &\Tilde{\theta}_{n+1}^l = \theta_n^l+\gamma_n H_l(X^l_{n+1}, \theta^l_n)\\
    & \theta_{n+1}^l = \begin{cases}
        \Tilde{\theta}_{n+1}^l, \quad|\Tilde{\theta}_{n+1}^l-\theta_n^l| < \epsilon_n\;\textrm{  and  }\; \theta^l_{n+1}\in \Theta_{n+1}\\
        \theta_0, \quad \textrm{otherwise}
    \end{cases}
\end{split}
\end{equation}
where $(\theta_0,X_0)\in\Theta_0\times\left(\mathbb{R}^d\right)^{T\Delta^{-1}_l+1}$ is an arbitrary initial pair. Under appropriate conditions the sequence $\{\theta_n^l\}_{n\in\mathbb{N}_0}$ will converge almost surely to a root of $\nabla_{\theta}\log p_{\theta}^l(y_{1:T})$.
Theorem 5.5 in \cite{andr3} provides us with these appropriate conditions to guarantee the almost sure convergence of these iterations for every $l\in\mathbb{N}_0$. Below we will state the theorem in the context used in this paper along with the assumptions we impose in our paper which we will use to verify the conditions of the theorem. We will also explain how the restated theorem relates to the original theorem given our assumptions.

\subsection{Assumptions and Main Theorem}\label{app_sec:ass}

For a measurable space $S$ and a measurable function $V:S\to[1,\infty)$ we define the operator $\|.\|_V$ on the space of measurable functions $f:S\to\mathbb{R}$ by $\|f\|_V=\sup_{x}|\frac{f(x)}{V(x)}| $. We define the space of functions 
$$
\mathcal{L}_V = \left\{ f:S\to\mathbb{R} \textrm{  measurable}~:~ \sup_{x\in S} \frac{|f(x)|}{V(x)} <\infty\right\},
$$
and $\|.\|_V$ is a norm on $\mathcal{L}_V$. We impose the assumptions \hyperref[assump:A1]{(A1-3)} below on $\nabla_{\theta}\log p_{\theta}^l(y_{1:T})$ and $\Theta$ to guarantee that they are well-behaved.


\begin{hypA}\label{assump:A1}
The set $\Theta$ is bounded.
\end{hypA}


\begin{hypA}
\label{assump:A2}
For every $l\in\mathbb{N}_0$ the function $\theta\to \nabla_{\theta}\log p_{\theta}^l(y_{1:T})$ is twice continuously differentiable. Moreover there exists a unique root for $\nabla_{\theta}\log p_{\theta}^l(y_{1:T})$. This root is the unique maximizer of $\log p_{\theta}^l(y_{1:T})$ and we denote it by $\theta_{\star}^l$.
\end{hypA}

\begin{hypA}
\label{assump:A3}
For every $l\in\mathbb{N}_0$ there exists a constant $M$ such that
$$ \sup_{\bar{\theta}\in\partial \Theta}\limsup_{\theta\to\bar{\theta}} \log p_{\theta}^l(y_{1:T}) < M < \log p_{\theta^l_\star}^l(y_{1:T}).$$
\end{hypA}

The following is \cite[Theorem 5.5]{andr3}, slightly modified to our notation.
\begin{theorem}
\label{thm:main}
Let $l\in\mathbb{N}_0$ and consider the sequence $\{\theta_n^l\}_{n\in\mathbb{N}_0}$ defined by iteration \eqref{eq:msa}. Assume \hyperref[assump:A1]{(A1-3)}. Suppose that there exist a function $V_l : (\mathbb{R}^d)^{T\Delta_l^{-1}+1} \rightarrow [1, \infty)$ and constants $p \geq 2$, $\beta\in (0, 1]$, $\lambda,\delta\in(0,1)$, $b_l>0$, a non-empty set $\mathsf{C}$, and a measure $\eta_l$ that satisfy the following:\\
(TA1) $\sup_{\theta\in\Theta} (K_{\theta,l}V_l^p)(x) \leq \lambda (V_l(x))^p + b \mathbbm{1}_{\mathsf{C}}(x)$.\\
(TA2) $\inf_{\theta\in\Theta} K_{\theta,l}(x,A)\geq \delta\eta_l(A) \quad \forall x\in\mathsf{C} \text{ and } A\in \mathcal{B}((\mathbb{R}^d)^{T\Delta_l^{-1}+1})$.\\
(TA3) There exists $C$ such that for all $x\in(\mathbb{R}^d)^{T\Delta_l^{-1}+1}$\textup{:}
$$ \sup_{\theta\in\Theta}|H_l(\theta,x)|\leq CV_l(x), $$
$$ \underset{\substack{(\theta,\theta')\in\Theta^2\\ \theta\not=\theta'}}{\sup} \|\theta-\theta'\|^{-\beta}|H_l(\theta,x)-H_l(\theta',x)|\leq CV_l(x) .$$
(TA4) There exists $C$ such that for all $(\theta,\theta')\in\Theta^2$\textup{:}
$$ \left \| K_{\theta,l}\varphi - K_{\theta',l}\varphi \right \|_{V_l} \leq C\left \|  \varphi\right \|_{V_l} |\theta-\theta'|^{\beta} \quad   \forall \varphi\in \mathcal{L}_{V_l},$$
$$ \left \| K_{\theta,l}\varphi - K_{\theta',l}\varphi \right \|_{V_l^p} \leq C\left \|  \varphi\right \|_{V_l^p} |\theta-\theta'|^{\beta} \quad \forall \varphi\in \mathcal{L}_{V_l^p}.$$
(AT5) There exists $\alpha\in(0,\beta)$ such that
$$\sum_n (\gamma_n^2 + \gamma_n\epsilon_n^{\alpha} + (\epsilon_n^{-1}\gamma_n)^{p})<\infty.$$
We have 
$$\theta_n^l\rightarrow \theta^l_\star \quad\textrm{ a.s.}$$
\end{theorem}

Theorem~\ref{thm:main} imposes requirements on the functions $K_{\theta,l}$, $H_l$, and the sequences $\{\gamma_n\}_{n\in\mathbb{N}_0}$, and $\{\epsilon_n\}_{n\in\mathbb{N}_0}$. To satisfy these condition we need more assumptions on the afore-mentioned functions and sequences, we state these additional assumptions below and denote them by \hyperref[assump:A4]{(A4-8)}. Let $\eta$ be a finite signed measure on a probability space. By the Hahn-Jordan decomposition there exist finite measures $\eta^+$ and $\eta^-$ such that $\eta=\eta^+-\eta^-$. Denote $|\eta| = \eta^++\eta^-$.

\begin{hypA}
\label{assump:A4}
For every $l\in\mathbb{N}_0$ there exist a function $W_l:\mathbb{R}^d\rightarrow[1,\infty)$, a pair $(\kappa,\rho)\in (0,1)\times\mathbb{R}^+$
and a set $\mathsf{C}\in\mathcal{B}(\mathbb{R}^d)$ such that for any $(x,\theta)\in\mathbb{R}^d\times\Theta$:
$$
Q_{\theta,l}(e^{W_l})(x) \leq \exp\{\kappa W_l(x)+\rho\mathbbm{1}_{\mathsf{C}}(x)\},
$$
where $W_l$ satisfies the growth condition
$$ \lim_{\|x\|\rightarrow \infty} \frac{W_l(x)}{\log(\|x\|)} = \infty.$$
\end{hypA}

\begin{hypA}
\label{assump:A5}
There a constant $C >0$ such that  $\frac{1}{C} \leq g_{\theta}(x,y_t) ,|\nabla_{\theta}g_{\theta}(x,y_t)|\leq C$ for all $1\leq t\leq T$, $(x,y,\theta)\in \mathbb{R}^{d}\times\mathbb{R}^{d_y}\times\Theta$.
\end{hypA}

\begin{hypA}
\label{assump:A6}
For every $l\in\mathbb{N}_0$  there exists $C\in\mathbb{R}$ such that for every $\theta\in\Theta : \mu_{\theta}\left( \exp\{W_l\}\right) < C$.
\end{hypA}

\begin{hypA}
\label{assump:A7}
For every $l\in\mathbb{N}_0$ there exist $\zeta >0$ and $C>0$ such that for every $(\theta,\theta',x)\in\Theta\times\Theta\times \mathbb{R}^d$:
$$\left| Q_{\theta,l}- Q_{\theta',l}\right|(\exp\{W_l\})(x) \leq C\|\theta-\theta'\|^{\zeta} \exp\{W_l(x)\}$$
and
$$ \left|\mu_{\theta}-\mu_{\theta'} \right|(\exp\{W_l\})\leq C\|\theta-\theta'\|^{\zeta}.$$
\end{hypA}

\begin{hypA}
\label{assump:A8}
For every $l\in\mathbb{N}_0$ there are finite non-zero measures $\Xi_l$ and $\Psi$ such that for every $(\theta,x,A)\in\Theta\times \mathbb{R}^d\times\mathcal{B}(\mathbb{R}^d)$:
$$  \mu_{\theta}(A)\geq  \Psi(A) $$
and
$$ Q_{\theta,l}(x,A)\geq \Xi_l(A). $$
\end{hypA}

\begin{hypA}
\label{assump:A9}
For every $l\in\mathbb{N}_0$ and $x\in(\mathbb{R}^d)^{T\Delta_l^{-1}+1}$ the function $\theta\mapsto H_l(\theta,x)$ is differentiable. Moreover there exist constants $C,q>0$ such that for every $(\theta,x)\in\Theta\times (\mathbb{R}^d)^{T\Delta_l^{-1}+1}:$
$$ |H_l(\theta,x)| < C(1 + |x|^q) ,$$
$$ \|\nabla_\theta H_l(\theta,x)\| < C(1 + |x|^q) .$$
\end{hypA}

In \cite{andr3} the authors state four conditions which they named (A1)-(A4). Here we breifly show that these four conditions follow from the theorem assumptions (TA1-4) and our assumptions \hyperref[assump:A1]{(A1-3)}. The notation $h(\theta)$ in \cite{andr3} is the negative of the gradient of the log likelihood in our case $-\nabla \log p_{\theta}^l(y_{1:T})$. (A1) in \cite{andr3} follows from our assumptions \hyperref[assump:A1]{(A1-3)}. Indeed, the function $-\log p_{\theta}^l(y_{1:T})$ is lower bounded and our \hyperref[assump:A3]{(A3)} assumes that it is continuously twice differentiable which allows the choice $w=-\log p_{\theta}^l(y_{1:T})$ as remarked by the authors of \cite{andr3}. Assume our \hyperref[assump:A2]{(A2)} and \hyperref[assump:A3]{(A3)} and choose $M_0=\frac{1}{2}M_1=\frac{1}{2}M$. This choice with the boundedness of $\Theta$ and the uniqueness of the root of $\nabla_{\theta}\log p_{\theta}^l(y_{1:T})$ implies the assumptions (A1)-(i),(iii),(iii),(iv) in \cite{andr3}. In \cite{andr3} the authors introduce conditions (DRI)-1,2,3 which implies (A2) and (A3) in their paper. Assumption (DRI1) is just our theorem's assumptions (TA1) and (TA2) but with $x\in (\mathbb{R}^d)^{T\Delta_l^{-1}+1}$ instead of $x\in\mathsf{C}$ in (TA2). Assumptions (DRI2) and (DRI3) are (TA3) and (TA4) in Theorem~\ref{thm:main} above. (A4) in \cite{andr3} is really (TA5) in Theorem~\ref{thm:main}, but $\alpha$ can be any value in the interval $(0,\beta)$ because of the stronger assumption (DRI).

As remarked in \cite{andr3} an easy choice of $\{\gamma_n\}_{n\in\mathbb{N}}$ and $\{\epsilon_n\}_{n\in\mathbb{N}}$ to satisfy (TA5) is to choose $\delta\in\left(1,\frac{1+\alpha}{1+\alpha/p}\right)$ and $\eta\in\left( \frac{\delta-1}{\alpha} , 1-\frac{\delta}{p}\right)$, then take $\{\gamma_n\}_{n\in\mathbb{N}}$ such that $\sum_n \gamma_n^{\delta}<\infty$ and $\epsilon_n=\mathcal{L}\gamma_n^{\eta}$ for any constant $\mathcal{L}>0$. Therefore we only need to verify conditions (TA1)-(TA4).

\begin{theorem}\label{theo:you_idiot}
    Assume \hyperref[assump:A1]{(A1-9)}. For every $l\in\mathbb{N}_0$ the conditions (TA1-4) of Theorem~\ref{thm:main} hold.
\end{theorem}

\subsection{Proofs for Theorem \ref{theo:you_idiot}}\label{app_sec:proof}

We will prove Theorem \ref{theo:you_idiot}  in a sequence of lemmata.
In general, for a transition kernel $P$ we say that a function $V$ is a drift function for $P$ if there exist constants $a\in(0,1),b>0$ and a small set $\mathsf{C}$ such that $PV\leq aV+b\mathbbm{1}_{\mathsf{C}}$. If $V$ rather satisfies $P\exp\{V\}\leq \exp\{aV+b\mathbbm{1}_{\mathsf{C}}\}$ we call it a multiplicative drift function. Assumption \hyperref[assump:A4]{(A4)} supposes that for every $l\in\mathbb{N}_0$ the existence of a common multiplicative drift function for the kernels $\{Q_{\theta,l}\}_{\theta\in\Theta}$. Lemma \ref{lm:M_multipicative_drift} below shows that this property can be extended to unit-step kernels $M_{\theta,l}$. Define the function $\mathcal{W}_l:(\mathbb{R}^d)^{\Delta_l^{-1}}\to\mathbb{R}$ by
$$\mathcal{W}_l(x)=
    \log(\sum_{k=1}^{\Delta_l^{-1}}\exp\{W_l(x_k)\}).
$$
we show that it is a multiplicative drift function for $M_{\theta,l}$. Multiplicative drift conditions have been considered in \cite{kont} and in particular for particle methods in \cite{delm,whiteley}.

\begin{lemma}\label{lm:M_multipicative_drift}
     Assume  \hyperref[assump:A4]{(A4)}. Then for any $l\in\mathbb{N}_0$ there exist $(\xi,b)\in(0,1)\times\mathbb{R}$ and a set $\mathsf{C}$ such that for any $(\theta,x)\in\Theta\times(\mathbb{R}^d)^{\Delta_l^{-1}}$\textup{:}
$$
M_{\theta,l}(\exp\{\mathcal{W}_l\})(x) \leq \exp\{\xi\mathcal{W}_l(x) + b\mathbb{I}_{\mathsf{C}}(x)\}
$$
with $\xi=(1+\kappa)/2$.
\end{lemma}
\begin{proof}
For each $1\leq i\leq \Delta^{-1}_l$ let $p_i:(\mathbb{R}^d)^{\Delta_l^{-1}} \rightarrow\mathbb{R}^d$ be the $i$-th component projection.
    \begin{equation*}
        \begin{split}
            M_{\theta,l}\left(\exp\{{\mathcal{W}_l}\}\right)=\;&\sum_{i=1}^{\Delta_l^{-1}}M_{\theta,l}\left( \exp\{W_l\circ p_i\} \right). \\
        \end{split}
    \end{equation*}
For each $1\leq i\leq n$ we have 
$$ M_{\theta,l}\left( \exp\{W_l\circ p_i \} \right) = \left(Q_{\theta,l}\right)^{\Delta_l^{-1}} \left(\exp\{W_l\circ p_i \}\right) = \left(Q_{\theta,l}\right)^{i} \left(\exp\{W_l\circ p_i \}\right)$$
where the powers associated to the operator $Q_{\theta,l}$ mean repeated composition (or iteration). Let $i>1$, by assumption \hyperref[assump:A4]{(A4)}
$$ 
\left(Q_{\theta,l}\right)^{i} \left(\exp\{W_l\circ p_i \}\right) \leq  \left(Q_{\theta,l}\right)^{i-1} \left(\exp\{\kappa W_l\circ p_{i-1} + \rho \}\right) \leq e^{\rho}\left(Q_{\theta,l}\right)^{i-1} \left(\exp\{W_l\circ p_{i-1} \}\right).
$$
Let $x\in(\mathbb{R}^d)^{\Delta^{-1}_l}$, for $i=1$ we have
$$ Q_{\theta,l}\left(\exp\{W_l\circ p_{1} \}\right)(x) \leq \exp\{\kappa W(p_{\Delta^{-1}_l}(x))+\rho\}\leq e^{\rho}\exp\{ \kappa\mathcal{W}_l(x) \}.$$
It is clear inductively that there exists a constant $A$ such that for every $(\theta,x)\in\Theta\times(\mathbb{R}^d)^{\Delta_l^{-1}}$
$$M_{\theta,l}\left( \exp\{\mathcal{W}_l \} \right)(x)\leq A \exp\{ \kappa\mathcal{W}_l(x)\}. $$
Let $\mathsf{C}=\left\{ x\in(\mathbb{R}^d)^{\Delta_l^{-1}}~:~ \mathcal{W}_l(x)<\frac{2\log(A)}{1-\kappa}+1 \right\}$, $b=\log(A)$, and $\xi=\frac{1+\kappa}{2}$ where $A$ is chosen large enough to guarantee that the set $\mathsf{C}$ is non-empty. With this choice we have that for every $(\theta,x)\in\Theta\times(\mathbb{R}^d)^{\Delta_l^{-1}}$
$$ M_{\theta,l}\left( \exp\{\mathcal{W}_l \} \right)(x)\leq A \exp\{ \kappa\mathcal{W}_l(x)\} \leq \exp\{\xi\mathcal{W}_l(x)+b\mathbbm{1}_{\mathsf{C}}(x) \}. $$
\end{proof}

\begin{remark}
    The definition of $M_{\theta,l}(x,dz)$ depends only on the last component of $x$. If we consider $x$ as an element of $\mathbb{R}^d$ (as discussed above) then the following the proof of lemma \ref{lm:M_multipicative_drift} shows that there exists a small set $\mathsf{C}\subset\mathbb{R}^d$, such that the following inequality holds 
    $$ M_{\theta,l}\left( \exp\{\mathcal{W}_l \} \right)(x)\leq \exp\{\xi W_l(x)+b\mathbbm{1}_{\mathsf{C}}(x) \}$$
    with the same values for $\xi$ and $b$ chosen in Lemma \ref{lm:M_multipicative_drift}.
\end{remark}

We are now in a position to prove that condition (TA1) is satisfied. The function $V_l:(\mathbb{R}^d)^{T\Delta_l^{-1}+1}\to\mathbb{R}$ defined by
$$V_l(x) = \left( \frac{1}{T+1}\exp\{W_1(x_0)\}+\frac{1}{T+1}\sum_{t=1}^T \exp\{\mathcal{W}_l(x_t)\}\right)^{1/2}$$
where $x_t\in(\mathbb{R}^d)^{\Delta_l^{-1}}$ for every $1\leq t\leq T$ and $x_0\in\mathbb{R}^d$. The function $V_l$ is our candidate function that will satisfy the condition (TA1) with $p=2$. Many of the constants in the statements of the Lemmata below depend upon $T$, but that dependence is never recorded as it is not needed in our analysis.
\begin{lemma}\label{lm:V_l_drift}
Assume \hyperref[assump:A4]{(A4-6)}. For every $l\in\mathbb{N}_0$ there exists $(\lambda,b,m)\in(0,1)\times\mathbb{R}^+\times\mathbb{R}^+$ such that $(\theta,z)\in\Theta\times(\mathbb{R}^d)^{T\Delta_l^{-1}+1}$\textup{:}
$$ (K_{\theta,l}V_l^2)(z)\leq \lambda (V_l(z))^2 + b\mathbbm{1}_{\mathsf{C}_{m}}(z) $$
where $\mathsf{C}_{m}=\left\{ z\in (\mathbb{R}^d)^{T\Delta_l^{-1}+1} | V_l(z) < m \right\}$.
\end{lemma}
\begin{proof}
    Let $\mathcal{V}_l=\exp\{\mathcal{W}_l\}$ and $\mathcal{V}_{l,0}=\exp\{W_l\}$. We have $\mathcal{V}_l, \mathcal{V}_{l,0}\geq1$, from assumption \hyperref[assump:A6]{(A6)} the function $\mathcal{V}_{l,0}$ is $\mu_{\theta}$-integrable for every $\theta\in\Theta$, and from Lemma \ref{lm:M_multipicative_drift} we have 
  that there exist  
  constants $C>1$ and $0<\xi<1$ such that for every $\theta\in\Theta$, $x\in(\mathbb{R}^d)^{\Delta^{-1}}$, and $w\in\mathbb{R}^d$ the inequalities $M_{\theta,l}(\mathcal{V}_l)(x)\leq C(\mathcal{V}_l(x))^{\xi}$ and $M_{\theta,l}(\mathcal{V}_l)(w)\leq C(\mathcal{V}_{l,0}(w))^{\xi}$ hold. Define the function $\bar{\mathcal{V}_l}:(\mathbb{R}^d)^{T\Delta_l^{-1}+1}\rightarrow[T+1,\infty)$ by $\bar{\mathcal{V}_l}(x) = \mathcal{V}_{l,0}(x_0)+ \sum_{t=1}^{T}\mathcal{V}_l(x_t)$,
    we will show that there exists $(\lambda,\bar{b},m)\in\mathbb{R}^3$ such that for every $(\theta,x)\in\Theta\times(\mathbb{R}^d)^{T\Delta^{-1}_l+1}$:
    $$ (K_{\theta,l}\bar{\mathcal{V}_l})(x)\leq \lambda \bar{\mathcal{V}_l}(x) + \bar{b}\mathbbm{1}_{\bar{\mathsf{C}}_{\bar{m}}}(x) $$
    where $\bar{\mathsf{C}}_{\bar{m}}=\left\{ x\in (\mathbb{R}^d)^{T\Delta_l^{-1}+1} | \bar{\mathcal{V}_l}(x) < \bar{m} \right\}$. Proving this inequality is equivalent to proving the lemma.
    
    We calculate the action of $K_{\theta,l}$ on the functions $x\mapsto \mathcal{V}_l(x_i)$ for every $1\leq i\leq T$ and $x\mapsto \mathcal{V}_{l,0}(x_0)$. 
    Let $0\leq i\leq T$ define the sets $S_i = \{(F_T,A_{1:T-1}^{1:N-1}) | F_i = N \}\subset\{1,...,N\}^{1+(T-1)(N-1)}$. The sets $S_i$ are non-empty because $F_T=N$ implies $F_t=N$ for all $0\leq t<T$. Moreover the sets $S_i$ are strict subsets of $\{1,...,N\}^{1+(T-1)(N-1)}$. We begin by observing the following bounds. For $0\leq i\leq T$ and $\theta\in\Theta$:
$$
\sum_{(F_T, A_{1:T-1}^{1:N-1})\in S_i}\int \delta_z(dx^N_{0:T})\frac{G_{\theta,T}^l(x^{F_T}_T)}{\sum_{k=1}^N G_{\theta,T}^l(x^Q_T)}
\prod_{j=1}^{N-1}\mu_{\theta}(dx_0^{j})M_{\theta,l}(x^j_0,dx^j_1)\prod_{t=2}^{T}\frac{G_{\theta,t-1}^l(x^{A^j_{t-1}}_{t-1})}{\sum_{k=1}^N G_{\theta,t-1}^l(x^Q_{t-1})}M_{\theta,l}(x^{A_{t-1}^j}_{t-1},dx^j_t)
$$
\begin{eqnarray}
& = & 1 - \sum_{(F_T, A_{1:T-1}^{1:N-1})\not\in S_i}\int \delta_z(dx^N_{0:T})
\frac{G_{\theta,T}^l(x^{F_T}_T)}{\sum_{k=1}^N G_{\theta,T}^l(x^Q_T)}
\prod_{j=1}^{N-1}\mu_{\theta}(dx_0^{j})M_{\theta,l}(x^j_0,dx^j_1)
\nonumber\\
& & \times
\prod_{t=2}^{T}\frac{G_{\theta,t-1}^l(x^{A^j_{t-1}}_{t-1})}{\sum_{k=1}^N G_{\theta,t-1}^l(x^Q_{t-1})}
M_{\theta,l}(x^{A_{t-1}^j}_{t-1},dx^j_t) \label{eq:V_l_drift_bound_1}\\
& \leq & 1 - (|\{1,\dots,N\}^{1+(T-1)(N-1)}\backslash S_i|)\left(\frac{1} {C^2N}\right)^{1+(T-1)(N-1)} \nonumber \\
 & \leq & 1 - \left(\frac{1} {C^2N}\right)^{1+(T-1)(N-1)} \nonumber
\end{eqnarray}    
where we used \hyperref[assump:A5]{(A5)} to go from lines 2 to 3. From assumption \hyperref[assump:A6]{(A6)} we have that for every $\theta\in\Theta$, $z\in(\mathbb{R}^d)^{T\Delta_l^{-1}+1}$, and $F_0\not=N$:
\begin{equation}\label{eq:V_l_drift_bound_2}
\int  \delta_z(dx^N_{0:T})\prod_{j=1}^{N-1}\mu_{\theta}(dx_0^j)M_{\theta,l}(x^j_0,dx^j_1)\prod_{t=2}^i M_{\theta,l}(x^{A_{t-1}^j}_{t-1}, dx^j_t)\mathcal{V}_{l,0}(x_0^{F_0}) = \int \mu_{\theta}(dx_0^{F_0})\mathcal{V}_{l,0}(x_0^{F_0})\leq C.
\end{equation}
Using \eqref{eq:V_l_drift_bound_2} and assumption \hyperref[assump:A5]{(A5)} we can further bound
$$
 \sum_{(F_T, A_{1:T-1}^{1:N-1})\not\in S_i}\int \delta_z(dx^N_{0:T}) \frac{G_{\theta,T}^l(x^{F_T}_T)}{\sum_{k=1}^N G_{\theta,T}^l(x^Q_T)}
\prod_{j=1}^{N-1}\mu_{\theta}(dx_0^{j})M_{\theta,l}(x^j_0,dx^j_1)
\prod_{t=2}^{T}\frac{ G_{\theta,t-1}^l(x^{A^j_{t-1}}_{t-1}) }{\sum_{k=1}^N G_{\theta,t-1}^l(x^Q_{t-1})}M_{\theta,l}(x^{A_{t-1}^j}_{t-1},dx^j_t)\mathcal{V}_l(x_0^{F_0})
 $$
 \begin{eqnarray}
 & \leq & \left(\frac{C^2}{N}\right)^{1+(T-1)(N-1)}\sum_{(F_T, A_{1:T-1}^{1:N-1})\not\in S_i}  \int  \delta_z(dx^N_{0:T})
\prod_{j=1}^{N-1}\mu_{\theta}(dx_0^j)M_{\theta,l}(x^j_0,dx^j_1)
\nonumber\\
& &\times 
\prod_{t=2}^i M_{\theta,l}(x^{A_{t-1}^j}_{t-1}, dx^j_t)\mathcal{V}_l(x_0^{F_0})\label{eq:V_l_drift_bound_3}\\               
& \leq & \left(\frac{C^2}{N}\right)^{1+(T-1)(N-1)} N^{1+(T-1)(N-1)}C = C^{3+2(T-1)(N-1)}.\nonumber
\end{eqnarray}
Suppose $i=0$. For $\theta\in\Theta$ we have
$$
        \sum_{1\leq F_T,A_{1:T-1}^{1:N-1}\leq N}\int \delta_z(dx^N_{0:T})\frac{G_{\theta,T}^l(x^{F_T}_T)}{\sum_{k=1}^N G_{\theta,T}^l(x^Q_T)}
\prod_{j=1}^{N-1}\mu_{\theta}(dx_0^{j})M_{\theta,l}(x^j_0,dx^j_1)\prod_{t=2}^{T}\frac{G_{\theta,t-1}^l(x^{A^j_{t-1}}_{t-1})}{\sum_{k=1}^N G_{\theta,t-1}^l(x^Q_{t-1})}
$$
$$
\times 
M_{\theta,l}(x^{A_{t-1}^j}_{t-1},dx^j_t)\mathcal{V}_{l,0}(x_i^{F_i})
        $$
$$       
              =  \mathcal{V}_{l,0}(z_i)\sum_{(F_T, A_{1:T-1}^{1:N-1})\in S_i}\int \delta_z(dx^N_{0:T})\frac{G_{\theta,T}^l(x^{F_T}_T)}{\sum_{k=1}^N G_{\theta,T}^l(x^Q_T)}
\prod_{j=1}^{N-1}\mu_{\theta}(dx_0^{j})M_{\theta,l}(x^j_0,dx^j_1)\prod_{t=2}^{T}\frac{G_{\theta,t-1}^l(x^{A^j_{t-1}}_{t-1})}{\sum_{k=1}^N G_{\theta,t-1}^l(x^Q_{t-1})}
$$
  \begin{equation}\label{eq:lemma_1_decomposition_i_0}
\times M_{\theta,l}(x^{A_{t-1}^j}_{t-1},dx^j_t)
        \end{equation}
 $$ 
        + \sum_{(F_T, A_{1:T-1}^{1:N-1})\not\in S_i}\int \delta_z(dx^N_{0:T})\frac{G_{\theta,T}^l(x^{F_T}_T)}{\sum_{k=1}^N G_{\theta,T}^l(x^Q_T)}
\prod_{j=1}^{N-1}\mu_{\theta}(dx_0^{j})M_{\theta,l}(x^j_0,dx^j_1)\prod_{t=2}^{T}\frac{G_{\theta,t-1}^l(x^{A^j_{t-1}}_{t-1})}{\sum_{k=1}^N G_{\theta,t-1}^l(x^Q_{t-1})}
$$
$$
\times
M_{\theta,l}(x^{A_{t-1}^j}_{t-1},dx^j_t)\mathcal{V}_{l,0}(x_i^{F_i}).
$$ 
Using the bounds \eqref{eq:V_l_drift_bound_1} and \eqref{eq:V_l_drift_bound_3} we bound the quantities in \eqref{eq:lemma_1_decomposition_i_0} from above by
\begin{equation}\label{eq:V_l_drift_bound_5}
     1 - \left(\frac{1} {C^2N}\right)^{1+(T-1)(N-1)}\mathcal{V}_{l,0}(z_0) + C^{3+2(T-1)(N-1)}.
\end{equation}
Now suppose $1\leq i\leq T$. For $\theta\in\Theta$ the action of $K_{\theta,l}$ on $x\mapsto \mathcal{V}_l(x_i)$ is
$$
\sum_{1\leq F_T,A_{1:T-1}^{1:N-1}\leq N}\int \delta_z(dx^N_{0:T})\frac{G_{\theta,T}^l(x^{F_T}_T)}{\sum_{k=1}^N G_{\theta,T}^l(x^Q_T)}
\prod_{j=1}^{N-1}\mu_{\theta}(dx_0^{j})M_{\theta,l}(x^j_0,dx^j_1)\prod_{t=2}^{T}\frac{G_{\theta,t-1}^l(x^{A^j_{t-1}}_{t-1})}{\sum_{k=1}^N G_{\theta,t-1}^l(x^Q_{t-1})}
$$
$$
\times
M_{\theta,l}(x^{A_{t-1}^j}_{t-1},dx^j_t)\mathcal{V}_l(x_i^{F_i})        
$$
 \begin{eqnarray}
        &= & \mathcal{V}_l(z_i)\sum_{(F_T, A_{1:T-1}^{1:N-1})\in S_i}\int \delta_z(dx^N_{0:T})\frac{G_{\theta,T}^l(x^{F_T}_T)}{\sum_{k=1}^N G_{\theta,T}^l(x^Q_T)}
        \prod_{j=1}^{N-1}\mu_{\theta}(dx_0^{j})M_{\theta,l}(x^j_0,dx^j_1)\prod_{t=2}^{T}\frac{G_{\theta,t-1}^l(x^{A^j_{t-1}}_{t-1})}{\sum_{k=1}^N G_{\theta,t-1}^l(x^Q_{t-1})}
        \nonumber \\
        & & \times M_{\theta,l}(x^{A_{t-1}^j}_{t-1},dx^j_t)\label{eq:lemma_1_decomposition_1}
 + \sum_{(F_T, A_{1:T-1}^{1:N-1})\not\in S_i}\int \delta_z(dx^N_{0:T})\frac{G_{\theta,T}^l(x^{F_T}_T)}{\sum_{k=1}^N G_{\theta,T}^l(x^Q_T)}
        \prod_{j=1}^{N-1}\mu_{\theta}(dx_0^{j})M_{\theta,l}(x^j_0,dx^j_1)\\ & & \times \prod_{t=2}^{T}\frac{G_{\theta,t-1}^l(x^{A^j_{t-1}}_{t-1})}{\sum_{k=1}^N G_{\theta,t-1}^l(x^Q_{t-1})}M_{\theta,l}(x^{A_{t-1}^j}_{t-1},dx^j_t)\mathcal{V}_l(x_i^{F_i}). \nonumber
\end{eqnarray}
From \eqref{eq:V_l_drift_bound_1} the first term in \eqref{eq:lemma_1_decomposition_1} is bounded above by $1 - \left(\frac{1} {C^2N}\right)^{1+(T-1)(N-1)}$.
For the second term in \eqref{eq:lemma_1_decomposition_1} we use the uniform boundedness of $G_{\theta,T}^l$
$$
      \sum_{(F_T, A_{1:T-1}^{1:N-1})\not\in S_i}\int \delta_z(dx^N_{0:T}) \frac{G_{\theta,T}^l(x^{F_T}_T)}{\sum_{k=1}^N G_{\theta,T}^l(x^Q_T)}
\prod_{j=1}^{N-1}\mu_{\theta}(dx_0^{j})M_{\theta,l}(x^j_0,dx^j_1)\prod_{t=2}^{T}\frac{G_{\theta,t-1}^l(x^{A^j_{t-1}}_{t-1})}{\sum_{k=1}^N G_{\theta,t-1}^l(x^Q_{t-1})}M_{\theta,l}(x^{A_{t-1}^j}_{t-1},dx^j_t)\mathcal{V}_l(x_i^{F_i})
$$
$$
       \leq \left(\frac{C^2}{N}\right)^{1+(T-1)(N-1)}\sum_{(F_T, A_{1:T-1}^{1:N-1})\not\in S_i}  \int  \delta_z(dx^N_{0:T})
\prod_{j=1}^{N-1}\mu_{\theta}(dx_0^j)M_{\theta,l}(x^j_0,dx^j_1)\prod_{t=2}^i M_{\theta,l}(x^{A_{t-1}^j}_{t-1}, dx^j_t)\mathcal{V}_l(x_i^{F_i}).
$$
Using the multiplicative drift property of $\mathcal{V}_l$ we have
\begin{equation}\label{eq:V_l_drift_bound_4}
    \begin{split}
      & \sum_{(F_T, A_{1:T-1}^{1:N-1})\not\in S_i}  \int \delta_z(dx^N_{0:T}) \prod_{j=1}^{N-1}\mu_{\theta}(dx_0^j)M_{\theta,l}(x^j_0,dx^j_1)\prod_{t=2}^i M_{\theta,l}(x^{A_{t-1}^j}_{t-1}, dx^j_t)\mathcal{V}_l(x_i^{F_i})\\
      & \leq C\sum_{(F_T, A_{1:T-1}^{1:N-1})\not\in S_i} \int  \delta_z(dx^N_{0:T})\prod_{j=1}^{N-1}\mu_{\theta}(dx_0^j)M_{\theta,l}(x^j_0,dx^j_1)\prod_{t=2}^{i-1} M_{\theta,l}(x^{A_{t-1}^j}_{t-1}, dx^j_t)(\mathcal{V}_l(x_{i-1}^{F_{i-1}}))^{\xi}\mathbbm{1}_{i>1}\\
      &+C\sum_{(F_T, A_{1:T-1}^{1:N-1})\not\in S_i} \int  \delta_z(dx^N_{0:T})\prod_{j=1}^{N-1}\mu_{\theta}(dx_0^j)M_{\theta,l}(x^j_0,dx^j_1)(\mathcal{V}_{l,0}(x_{i-1}^{F_{i-1}}))^{\xi}\mathbbm{1}_{i=1}
    \end{split}
\end{equation}
We consider the cases $F_{i-1} = N$ and $F_{i-1} \not= N$. Suppose $F_{i-1} = N$. The quantities in \eqref{eq:V_l_drift_bound_4} are bounded above by
\begin{equation*}
    \begin{split}
& C\sum_{(F_T, A_{1:T-1}^{1:N-1})\not\in S_i} \max((\mathcal{V}_{l,0}(z_{0}))^{\xi},(\mathcal{V}_l(z_{i-1}))^{\xi})
~\leq~CN^{1+(T-1)(N-1)}(\bar{\mathcal{V}}_l(z))^{\xi}.
    \end{split}
\end{equation*}
For the case $F_{i-1} \not = N$. Using $\mathcal{V}_l,\mathcal{V}_{l,0}\geq1$ and \eqref{eq:V_l_drift_bound_2} we can bound the quantities in \eqref{eq:V_l_drift_bound_4} by
\begin{equation*}
    \begin{split}
      & C\sum_{(F_T, A_{1:T-1}^{1:N-1})\not\in S_{i-1}} \int  \delta_z(dx^N_{0:T})\prod_{j=1}^{N-1}\mu_{\theta}(dx_0^j)M_{\theta,l}(x^j_0,dx^j_1)\prod_{t=2}^{i-1} M_{\theta,l}(x^{A_{t-1}^j}_{t-1}, dx^j_t)\mathcal{V}_l(x_{i-1}^{F_{i-1}})\mathbbm{1}_{i>1}\\
      +\;&C\sum_{(F_T, A_{1:T-1}^{1:N-1})\not\in S_{i-1}} \int  \delta_z(dx^N_{0:T})\prod_{j=1}^{N-1}\mu_{\theta}(dx_0^j)M_{\theta,l}(x^j_0,dx^j_1)\mathcal{V}_{l,0}(x_{i-1}^{F_{i-1}})\mathbbm{1}_{i=1}\\
      \leq\; & C\sum_{(F_T, A_{1:T-1}^{1:N-1})\not\in S_{i-1}} \int  \delta_z(dx^N_{0:T})\prod_{j=1}^{N-1}\mu_{\theta}(dx_0^j)M_{\theta,l}(x^j_0,dx^j_1)\prod_{t=2}^{i-1} M_{\theta,l}(x^{A_{t-1}^j}_{t-1}, dx^j_t)\mathcal{V}_l(x_{i-1}^{F_{i-1}})\\
      +\;&C^2N^{1+(T-1)(N-1)}.
    \end{split}
\end{equation*}
Therefore we have shown that
$$
\sum_{(F_T, A_{1:T-1}^{1:N-1})\not\in S_i}  \int \delta_z(dx^N_{0:T}) \prod_{j=1}^{N-1}\mu_{\theta}(dx_0^j)M_{\theta,l}(x^j_0,dx^j_1)\prod_{t=2}^i M_{\theta,l}(x^{A_{t-1}^j}_{t-1}, dx^j_t)\mathcal{V}_l(x_i^{F_i}) 
$$
\begin{eqnarray*}
      & \leq &  C\left\{\bar{\mathcal{V}}_l(z)^{\xi} + C^2\right\}N^{1+(T-1)(N-1)}
      +  C\sum_{(F_T, A_{1:T-1}^{1:N-1})\not\in S_{i-1}} \int  \delta_z(dx^N_{0:T})\prod_{j=1}^{N-1}\mu_{\theta}(dx_0^j)M_{\theta,l}(x^j_0,dx^j_1)\\ & & \prod_{t=2}^{i-1} M_{\theta,l}(x^{A_{t-1}^j}_{t-1}, dx^j_t)\mathcal{V}_l(x_{i-1}^{F_{i-1}}).
\end{eqnarray*}
Hence by induction, the bound \eqref{eq:V_l_drift_bound_3}, and the fact that $\mathcal{V}_l,\mathcal{V}_{l,0}\geq 1$ we have shown that there is a constant $C$ such that for every $0\leq i\leq T$ and $z\in(\mathbb{R}^d)^{T\Delta_l^{-1}+1}$:
\begin{equation}\label{eq:V_l_drift_bound_6}
\sum_{(F_T, A_{1:T-1}^{1:N-1})\not\in S_i}  \int \delta_z(dx^N_{0:T}) \prod_{j=1}^{N-1}\mu_{\theta}(dx_0^j)M_{\theta,l}(x^j_0,dx^j_1)\prod_{t=2}^i M_{\theta,l}(x^{A_{t-1}^j}_{t-1}, dx^j_t)\mathcal{V}_l(x_i^{F_i}) \leq  C\bar{\mathcal{V}}_l(z)^{\xi}.
\end{equation}
Let $C_1>\max\left( \left(C^2N\right)^{1+(T-1)(N-1)}, C(T+1)\right)$.
Combining the bounds \eqref{eq:V_l_drift_bound_1} and \eqref{eq:V_l_drift_bound_6} we have
$$
            (K_{\theta,l}\bar{\mathcal{V}_l})(z) 
            \leq \left(1-\frac{1}{C_1}\right)\mathcal{V}_{l,0}(z_0) + \sum_{i=1}^T \left(\left(1-\frac{1}{C_1}\right)\mathcal{V}_l(z_i) + \frac{C_1}{T+1}(\bar{\mathcal{V}_l}(z))^{\xi}\right)
            = \left(1-\frac{1}{C_1}\right)\bar{\mathcal{V}_l}(z) + C_1(\bar{\mathcal{V}_l}(z))^{\xi}.
$$
Let $\lambda = \left(1-\frac{1}{2C_1}\right)$, $\bar{m} > \left(2C_1^2\right)^{\frac{1}{1-\xi}}$, $\bar{b}> C_1(2C_1^2)^{\frac{\xi}{1-\xi}}$, and the set $\mathsf{C}_{\bar{m}}=\{x\in(\mathbb{R}^d)^{\Delta_l^{-1}+1}\;:\;\bar{\mathcal{V}_l}(x) < \bar{m}\}$. We have
\begin{equation*}
    \left(K_{\theta,l}\bar{\mathcal{V}_l}\right)(z) \leq \lambda \bar{\mathcal{V}_l}(z) + \bar{b}\mathbbm{1}_{\bar{\mathsf{C}}_{\bar{m}}}(z).
\end{equation*}
Finally taking $m=\sqrt{\frac{\bar{m}}{T+1}}$,$b=\frac{\bar{b}}{T+1}$, and $\mathsf{C}_m = \{x\in(\mathbb{R}^d)^{\Delta_l^{-1}+1}\;:\;V_l(x) < m\}$ then
\begin{equation*}
    \left(K_{\theta,l}V_l^2\right)(z) \leq \lambda V_l^2(z) + b\mathbbm{1}_{\mathsf{C}_{m}}(z).
\end{equation*}
The constant $C_1$ can be taken as large as needed to guarantee that the set $\mathsf{C}_m$ is non-empty. Furthermore, the constant $C_1$ is independent of $\theta$ and $z$; consequently $\lambda$, $m$, and $b$ are also independent of $\theta$ and $z$ which completes the proof.
\end{proof}

\begin{corollary}\label{cl:finite_V_l_norm_1}
Assume \hyperref[assump:A4]{(A4-6)}. For every $l\in\mathbb{N}_0$, $A^{1:N-1}_{1:T-1}\in\{1,...,N\}^{(T-1)(N-1)}$, $F_T\in\{1,...,N\}$, and $r\in\{1,2\}$:
$$ \sup_{z\in(\mathbb{R}^d)^{T\Delta_l^{-1}+1}}\sup_{\theta\in\Theta}\left\| \int\delta_{z}(dx_{0:T}^N)\prod_{j=1}^{N-1}\mu_{\theta}(dx_0^{j})M_{\theta,l}(x^j_0,dx^j_1)\prod_{t=2}^{T}M_{\theta,l}(x^{A_{t-1}^j}_{t-1},dx^j_t)V_l^r(x_0^{F_0},...,x_T^{F_T}) \right\|_{V_l^r} < \infty. $$

\end{corollary}
\begin{proof}
  As Lemma \ref{lm:V_l_drift} holds, we have
    \begin{equation*}
        \begin{split}
            &\int\delta_{z}(dx_{0:T}^N)\prod_{j=1}^{N-1}\mu_{\theta}(dx_0^{j})M_{\theta,l}(x^j_0,dx^j_1)\prod_{t=2}^{T}M_{\theta,l}(x^{A_{t-1}^j}_{t-1},dx^j_t)V_l^2(x_0^{F_0},...,x_T^{F_T}) \\\leq\;& C^{1+(T-1)(N-1)}(K_{\theta,l}V^2_l)(z) \\
            \leq\;& C^{1+(T-1)(N-1)}(\lambda V^2_l(z) + b\mathbbm{1}_{\mathsf{C}_d}(z))\\
            \leq\;& (\lambda C^{1+(T-1)(N-1)} + m) V_l^2(z)
        \end{split}
    \end{equation*}
    where we used \hyperref[assump:A5]{(A5)} to and the definition of $K_{\theta,l}$ to move to the second line. This proves the inequality for $r=2$. To prove it for $r=1$ we will prove that $V$ itself is a drift function. Using Jensen's inequality we have
    $$ (PV)(z) \leq \sqrt{(PV^2)(z)} \leq \sqrt{\lambda V^2(z) + b\mathbbm{1}_{\mathsf{C}_d}}(z)  \leq \sqrt{\lambda}V(z) + \sqrt{b}\mathbbm{1}_{\mathsf{C}_d}(z).$$
\end{proof}

The proof of Lemma \ref{lm:V_l_drift} did not require the $M$'s, $\mu$'s and $G$'s to have the same $\theta$ subscripts. This observation allows immediately to generalize Corollary \ref{cl:finite_V_l_norm_1} to mixed $\theta$'s subscripts. This will be helpful in verifying the theorem condition (TA4).

\begin{corollary}\label{cl:finite_V_l_norm_2}
    Assume \hyperref[assump:A4]{(A4-6)}. For every $l\in\mathbb{N}_0$, $A^{1:N-1}_{1:T-1}\in\{1,...,N\}^{(T-1)(N-1)}$, $F_T\in\{1,...,N\}$, and $r\in\{1,2\}$. There exists a constant $C$ such that for every $\vartheta\in\Theta^{(T+1)\times (N-1)}$ and $z\in(\mathbb{R}^d)^{T\Delta_l^{-1}+1}$:
    $$ \left\| \int\delta_{z}(dx_{0:T}^N)\prod_{j=1}^{N-1}\mu_{\vartheta_{1,j}}(dx_0^{j})M_{\vartheta_{2,j},l}(x^j_0,dx^j_1)\prod_{t=2}^{T}M_{\vartheta_{t+1,j},l}(x^{A_{t-1}^j}_{t-1},dx^j_t)V_l^r(x_0^{F_0},...,x_T^{F_T}) \right\|_{V_l^r} < C. $$
\end{corollary}
We have proven that (TA1) holds. 
\begin{lemma}
    Assume \hyperref[assump:A5]{(A5)} and \hyperref[assump:A8]{(A8)}. For every $l\in\mathbb{N}_0$ the assumption (TA2) holds with $\eta_l$ defined above.
\end{lemma}
\begin{proof}
    Let $z\in (\mathbb{R}^d)^{T\Delta_l^{-1}+1}$ and a measurable set $S\subset\mathcal{B}((\mathbb{R}^d)^{T\Delta_l^{-1}+1})$. We have
\begin{equation}\label{eq:minorization_ineq_1}
    \begin{split}
        K_{\theta,l}(z,S) = \;&(K_{\theta,l}\mathbbm{1}_{S})(z)\\
        \geq\;& \left(\frac{1}{C^2N}\right)^{1+(T-1)(N-1)} \sum_{1\leq F_T, A_{1:T-1}^{1:N-1}\leq N}\int \delta_z(dx^N_{0:T})\prod_{j=1}^{N-1}\mu_{\theta}(dx_0^{j})M_{\theta,l}(x^{j}_{0},dx^j_1)\\&\prod_{t=2}^{T}M_{\theta,l}(x^{A_{t-1}^j}_{t-1},dx^j_t)\mathbbm{1}_S(x_0^{F_0},...,x_T^{F_T})\\
        \geq& \left(\frac{1}{C^2N}\right)^{1+(T-1)(N-1)}\int\mu_{\theta}(dx_0^{1})M_{\theta,l}(x^{1}_{0},dx^1_1)\prod_{t=2}^{T}M_{\theta,l}(x^{1}_{t-1},dx^1_t)\mathbbm{1}_S(x_0^{1},...,x_T^{1}).
    \end{split}
\end{equation}
Recall that $x_t\in(\mathbb{R}^d)$ for every $1\leq t\leq T$ and $x_0\in\mathbb{R}^d$. We write $x_{t,k}\in\mathbb{R}$ for the components of $x_t$ for $1\leq t\leq T$ and $1\leq k\leq \Delta_l^{-1}$. From assumption \hyperref[assump:A8]{(A8)} and the definition of $M_{\theta,l}$ we have
$$ M_{\theta,l}(x_{t-1}^1,dx_t^1) = Q_{\theta,l}(x_{t-1,\Delta_l^{-1}},dx^1_{t,1})\prod_{i=2}^{\Delta_l^{-1}} Q_{\theta,l}(x_{t,i-1},dx_{t,i}) \geq \prod_{i=1}^{\Delta_l^{-1}} \Xi_{l}(dx_{t,i}).$$
From assumption \hyperref[assump:A8]{(A8)} we also have
$$ \mu_{\theta}(dx_0^1)\geq \Psi(dx_0^1). $$
Therefore the last quantity in \eqref{eq:minorization_ineq_1} is bounded below by
$$ \left(\frac{1}{C^2N}\right)^{1+(T-1)(N-1)}\int\Psi(dx_0^1)\prod_{t=1}^{T}\prod_{i=1}^{\Delta_l^{-1}} \Xi_{l}(dx_{t,i})\mathbbm{1}_S(x_0^{1},...,x_T^{1}). $$
This measure is independent of $\theta$ and $z$ which proves the minorization condition (TA2).
\end{proof}

\begin{lemma}
    Assume \hyperref[assump:A1]{(A1)}, \hyperref[assump:A4]{(A4)}, and \hyperref[assump:A9]{(A9)}. Assumption (TA3) holds with $\beta=1$.
\end{lemma}

\begin{proof}
From assumption \hyperref[assump:A9]{(A9)} there are $q,C_1>0$ such that for every $(\theta,x)\in\Theta\times(\mathbb{R}^d)^{T\Delta_l^{-1}+1}$ the inequality $\max\left(\|\nabla_{\theta}H_l(\theta,x)\|,|H_l(\theta,x)|\right)<C_1(1+|x|^q)$ holds. From assumption \hyperref[assump:A4]{(A4)} and the definition of $V_l$ we have that $\lim_{\|x\|\rightarrow \infty} \log V_l(x)/\log\|x\|=\infty$. Hence there exists a constant $I>0$ such that $V_l(x) > 1+\|x\|^q$ for $\|x\|>I$. On the compact set $\|x\|\leq I$ the function $V_l(x)/(1+\|x\|^q)$ is continuous and positive hence has a minimum $m>0$. Let $C = C_1/\min(1,m)$. For every $(\theta,x)\in\Theta\times(\mathbb{R}^d)^{T\Delta_l^{-1}+1}$ we have
$$ CV_l(x) \geq C_1(1+\|x\|^q) > \max\left(\|\nabla_{\theta}H_l(\theta,x)\|,|H_l(\theta,x)|\right)$$
form which the first inequality of (TA3) follows. For the second inequality let $\theta,\theta'\in\Theta$ and $x\in (\mathbb{R}^d)^{T\Delta_l^{-1}+1}$. Applying mean value theorem, Cauchy-Schwarz inequality, and the previous inequality we have
$$|H_l(\theta,x)-H_l(\theta',x)|\leq \|\theta-\theta'\|\sup_{\theta\in\Theta}\|\nabla_\theta H_l(\theta,x)\|\ < CV_l(x)\|\theta-\theta'\|.$$
\end{proof}

\begin{lemma}
    Assume \hyperref[assump:A1]{(A1)} and \hyperref[assump:A4]{(A4-7)}. Assumption (TA4) holds with $\beta = \min(1,\zeta)$ and $p=2$.
\end{lemma}

\begin{proof}
From Lemma \ref{lm:V_l_drift} the assumption (TA1) holds with $p=2$. Let $r\in\{1,2\}$ and $\varphi\in\mathcal{L}_{V_l^{r}}$.
Denote
$$ G^l_{\theta}(x_{0:T}^{1:N}, F_T, A_{1:T-1}^{1:N-1}) =  \frac{G^l_{T,\theta}(x^{F_T}_T)}{\sum_{k=1}^N G^l_{T,\theta}(x^k_T)}\prod_{t=2}^{T}\frac{G^l_{t-1,\theta}(x^{A^j_{t-1}}_{t-1})}{\sum_{k=1}^N G^l_{t-1,\theta}(x^k_{t-1})}.$$
The functions $\{G^l_{\theta}\}_{\theta\in\Theta}$ and $\{\nabla_{\theta}G^l_{\theta}\}_{\theta\in\Theta}$ are uniformly bounded. For the ease of notations define $\mathcal{S}=\{(F_T,A^{1:N-1}_{1:T-1})\; :\; 1\leq F_T,A^{1:N-1}_{1:T-1}\leq N\}$. Define $A_0^j=j$ for $1\leq j\leq N$. We write the decomposition
\begin{eqnarray}
(K_{\theta,l}\varphi)(z) - (K_{\theta',l}\varphi)(z) & = & 
\sum_{\mathcal{S}} \int \left(G^l_{\theta}(x_{0:T}^{1:N}, F_T, A_{1:T-1}^{1:N-1}) - G^l_{\theta'}(x_{0:T}^{1:N}, F_T, A_{1:T-1}^{1:N-1})\right)\delta_{z}(dx_{0:T}^N) \nonumber\\
& &\times \prod_{j=1}^{N-1}\mu_{\theta}(dx_0^{j})\prod_{t=1}^{T}M_{\theta,l}(x^{A_{t-1}^j}_{t-1},dx^j_t)\varphi(x_0^{F_0},...,x_T^{F_T})\nonumber\\
& &+ \sum_{\mathcal{S}} \int G^l_{\theta'}(x_{0:T}^{1:N}, F_T, A_{1:T-1}^{1:N-1})\delta_{z}(dx_{0:T}^N)\nonumber
\end{eqnarray}
\begin{eqnarray}
 & &\times\left(\prod_{j=1}^{N-1}\mu_{\theta}(dx_0^{j})-\prod_{j=1}^{N-1}\mu_{\theta'}(dx_0^{j})\right)\prod_{j=1}^{N-1}\prod_{t=1}^{T}M_{\theta,l}(x^{A_{t-1}^j}_{t-1},dx^j_t)\varphi(x_0^{F_0},...,x_T^{F_T})\nonumber\\
& &
 +  \sum_{\mathcal{S}} \int G^l_{\theta'}(x_{0:T}^{1:N}, F_T, A_{1:T-1}^{1:N-1})\delta_{z}(dx_{0:T}^N) )\nonumber
\\
& & \times
\prod_{j=1}^{N-1}\mu_{\theta'}(dx_0^{j})\left(\prod_{t=1}^{T}M_{\theta,l}(x^{A_{t-1}^j}_{t-1},dx^j_t)-\prod_{t=1}^{T}M_{\theta'}(x^{A_{t-1}^j}_{t-1},dx^j_t)\right)\varphi(x_0^{F_0},...,x_T^{F_T}).\label{eq:collapsing_sum_1}
 \end{eqnarray}
Applying the mean value theorem to the function $G^l_{\theta}$ along with the uniform boundedness of its derivative from assumption \hyperref[assump:A5]{(A5)}
$$ |G^l_{\theta}(x_{0:T}^{1:N}, F_T, A_{1:T-1}^{1:N-1}) - G^l_{\theta'}(x_{0:T}^{1:N}, F_T, A_{1:T-1}^{1:N-1})| = |\nabla_{\theta} G_{\theta''}(x_{0:T}^{1:N}, F_T, A_{1:T-1}^{1:N-1})\cdot(\theta - \theta')| \leq C \|\theta-\theta'\|.$$
Hence first term is bounded from above by
\begin{equation*}
    \begin{split}
       & C\|\theta'-\theta'\|\int \delta_{z}(dx_{0:T}^N)\prod_{j=1}^{N-1}\mu_{\theta}(dx_0^{j})\prod_{t=1}^{T}M_{\theta,l}(x^{A_{t-1}^j}_{t-1},dx^j_t)\varphi(x_0^{F_0},...,x_T^{F_T})\\
       =\;& C\|\theta'-\theta'\|\int \delta_{z}(dx_{0:T}^N)\prod_{j=1}^{N-1}\mu_{\theta}(dx_0^{j})\prod_{t=1}^{T}M_{\theta,l}(x^{A_{t-1}^j}_{t-1},dx^j_t)\frac{\varphi(x_0^{F_0},...,x_T^{F_T})}{V_l^{r}(x_0^{F_0},...,x_T^{F_T})}V_l^{r}(x_0^{F_0},...,x_T^{F_T})\\
       \leq\;&  C\|\theta'-\theta'\|\|\varphi\|_{V_l^{r}}\int \delta_{z}(dx_{0:T}^N)\prod_{j=1}^{N-1}\mu_{\theta}(dx_0^{j})\prod_{t=1}^{T}M_{\theta,l}(x^{A_{t-1}^j}_{t-1},dx^j_t)V_l^{r}(x_0^{F_0},...,x_T^{F_T}).
    \end{split}
\end{equation*}
By Corollary \ref{cl:finite_V_l_norm_1} is bounded above by $CV^r(z)$. For the second term we can further decompose the term $\prod_{j=1}^{N-1}\mu_{\theta}(dx_0^{j})-\prod_{j=1}^{N-1}\mu_{\theta'}(dx_0^{j})$ as 
\begin{equation*}
    \begin{split}
        &\prod_{j=1}^{N-1}\mu_{\theta}(dx_0^{j})-\prod_{j=1}^{N-1}\mu_{\theta'}(dx_0^{j})\\
        =& \sum_{i=1}^{N-1} \left(\prod_{j=1}^{N-1-i}\mu_{\theta}(dx_0^{j})\right)(\mu_{\theta}(dx_0^{N-i}) - \mu_{\theta'}(dx_0^{N-i}))\left(\prod_{j=N-i+1}^{N-1}
        \mu_{\theta'}(dx_0^{j})\right)
    \end{split}
\end{equation*}
with the convention that an empty product is $1$. Similarly to the argument of the first term it is enough to bounded the second term for $\varphi=V_l^{r}$. Since $G$ is bounded it is sufficient to study the quantity
\begin{equation}\label{eq:collapsing_sum_2}
\begin{split}
        &\sum_{i=1}^{N-1}\sum_{\mathcal{S}} \int \delta_{z}(dx_{0:T}^N)\left(\prod_{j=1}^{N-1-i}\mu_{\theta}(dx_0^{j})\right)|\mu_{\theta}(dx_0^{N-i}) - \mu_{\theta'}(dx_0^{N-i})|\left(\prod_{j=N-i+1}^{N-1}\mu_{\theta'}(dx_0^{j})\right)\\
        &\quad \quad \quad\quad\quad\prod_{j=1}^{N-1}\prod_{t=1}^{T}M_{\theta,l}(x^{A_{t-1}^j}_{t-1},dx^j_t)V_l^{r}(x_0^{F_0},...,x_T^{F_T})\\
\end{split}
\end{equation}
where $|\mu_{\theta}(dx_0^{N-i}) - \mu_{\theta'}(dx_0^{N-i})|$ is, as defined above, the sum of the measures in Hahn-Jordan decomposition of the signed measure $\mu_{\theta}(dx_0^{N-i}) - \mu_{\theta'}(dx_0^{N-i})$. By the Cauchy-Schwarz inequality
$$
 \int \delta_{z}(dx_{0:T}^N)\left(\prod_{j=1}^{N-1-i}\mu_{\theta}(dx_0^{j})\right)|\mu_{\theta}(dx_0^{N-i}) - \mu_{\theta'}(dx_0^{N-i})|\left(\prod_{j=N-i+1}^{N-1}\mu_{\theta'}(dx_0^{j})\right)\times
 $$
 $$
  \prod_{j=1}^{N-1}\prod_{t=1}^{T}M_{\theta,l}(x^{A_{t-1}^j}_{t-1},dx^j_t)V_l(x_0^{F_0},...,x_T^{F_T}) 
  $$
   $$     
  \leq \left( \int \delta_{z}(dx_{0:T}^N)\left(\prod_{j=1}^{N-1-i}\mu_{\theta}(dx_0^{j})\right)|\mu_{\theta}(dx_0^{N-i}) - \mu_{\theta'}(dx_0^{N-i})|\left(\prod_{j=N-i+1}^{N-1}\mu_{\theta'}(dx_0^{j})\right)
        \prod_{j=1}^{N-1}\prod_{t=1}^{T}M_{\theta,l}(x^{A_{t-1}^j}_{t-1},dx^j_t)\right)^{\frac{1}{2}}\times
 $$
  $$      
 \Bigg( \int \delta_{z}(dx_{0:T}^N)\left(\prod_{j=1}^{N-1-i}\mu_{\theta}(dx_0^{j})\right)|\mu_{\theta}(dx_0^{N-i}) - \mu_{\theta'}(dx_0^{N-i})|\left(\prod_{j=N-i+1}^{N-1}\mu_{\theta'}(dx_0^{j})\right)\times
 $$
 $$
        \prod_{j=1}^{N-1}\prod_{t=1}^{T}M_{\theta,l}(x^{A_{t-1}^j}_{t-1},dx^j_t)V_l^2(x_0^{F_0},...,x_T^{F_T}) \Bigg)^{\frac{1}{2}}
 $$       
$$       
        = \left(\int \delta_{z}(dx_{0:T}^N)|\mu_{\theta}(dx_0^{N-i}) - \mu_{\theta'}(dx_0^{N-i})| \right)^{\frac{1}{2}}
        \Bigg(\int \delta_{z}(dx_{0:T}^N)\left(\prod_{j=1}^{N-1-i}\mu_{\theta}(dx_0^{j})\right)|\mu_{\theta}(dx_0^{N-i}) - \mu_{\theta'}(dx_0^{N-i})|\times
$$
$$        
        \left(\prod_{j=N-i+1}^{N-1}\mu_{\theta'}(dx_0^{j})\right)
        \prod_{j=1}^{N-1}\prod_{t=1}^{T}M_{\theta,l}(x^{A_{t-1}^j}_{t-1},dx^j_t)V_l^2(x_0^{F_0},...,x_T^{F_T})\Bigg)^{\frac{1}{2}}
$$
$$
        \leq \Bigg(\int \delta_{z}(dx_{0:T}^N) |\mu_{\theta}(dx_0^{N-i}) - \mu_{\theta'}(dx_0^{N-i})| \mathcal{V}_l (x_0^{N-i})\Bigg)^{\frac{1}{2}}
        \Bigg(\int \delta_{z}(dx_{0:T}^N)\Bigg(\prod_{j=1}^{N-1-i}\mu_{\theta}(dx_0^{j})\Bigg)\times
$$
$$        
|\mu_{\theta}(dx_0^{N-i}) - \mu_{\theta'}(dx_0^{N-i})|\Bigg(\prod_{j=N-i+1}^{N-1}\mu_{\theta'}(dx_0^{j})\Bigg)
\prod_{j=1}^{N-1}\prod_{t=1}^{T}M_{\theta,l}(x^{A_{t-1}^j}_{t-1},dx^j_t)V_l^2(x_0^{F_0},...,x_T^{F_T})\Bigg)^{\frac{1}{2}}
        $$
        $$
        \leq C\|\theta-\theta'\|^{\frac{\zeta}{2}} \Bigg(\int \delta_{z}(dx_{0:T}^N)\left(\prod_{j=1}^{N-1-i}\mu_{\theta}(dx_0^{j})\right)|\mu_{\theta}(dx_0^{N-i}) - \mu_{\theta'}(dx_0^{N-i})|\left(\prod_{j=N-i+1}^{N-1}\mu_{\theta'}(dx_0^{j})\right)\times
$$
$$
\prod_{j=1}^{N-1}\prod_{t=1}^{T}M_{\theta,l}(x^{A_{t-1}^j}_{t-1},dx^j_t)V_l^2(x_0^{F_0},...,x_T^{F_T})\Bigg)^{\frac{1}{2}}
$$
where we have used $V_l\geq 1$. Hence we only need to to consider $r=2$. $V_l^2(x)$  is defined by $\frac{1}{T+1}\sum_{k=0}^{T} \mathcal{V}_l(x_k)$. In manner exactly similar to the proof of Lemma \ref{lm:V_l_drift} we start with a $0\leq k\leq T$ and keep descending with the indices as we integrate w.r.t to $M_{\theta,l}$'s until some $F_r$ for some $0\leq r\leq k$ hits $N$, the index of the conditioned path $z$, or we end up with $F_0\not=N$. In the first case we can find $C_1>0$ such that
\begin{equation*}
    \begin{split}
        &\int \left(\prod_{j=1}^{N-1-i}\mu_{\theta}(dx_0^{j})\right)|\mu_{\theta}(dx_0^{N-i}) - \mu_{\theta'}(dx_0^{N-i})|\left(\prod_{j=N-i+1}^{N-1}\mu_{\theta'}(dx_0^{j})\right)
        \prod_{j=1}^{N-1}\prod_{t=1}^{T}M_{\theta,l}(x^{A_{t-1}^j}_{t-1},dx^j_t)\mathcal{V}_l(x_k^{F_k})\\
        \leq\; & C_1\mathcal{V}_l(z_r) \int \left(\prod_{j=1}^{N-1-i}\mu_{\theta}(dx_0^{j})\right)|\mu_{\theta}(dx_0^{N-i}) - \mu_{\theta'}(dx_0^{N-i})|\left(\prod_{j=N-i+1}^{N-1}\mu_{\theta'}(dx_0^{j})\right)
        \prod_{j=1}^{N-1}\prod_{t=1}^{r}M_{\theta,l}(x^{A_{t-1}^j}_{t-1},dx^j_t)\\
        =\; & C_1\mathcal{V}_l(z_r) \int |\mu_{\theta}(dx_0^{N-i}) - \mu_{\theta'}(dx_0^{N-i})|\\
        =\; & C_1\mathcal{V}_l(z_r) \int |\mu_{\theta}(dx_0^{N-i}) - \mu_{\theta'}(dx_0^{N-i})|\mathcal{V}_l(x_0^{N-i})\\
        \leq \;& C_2 \mathcal{V}_l(z_r)\|\theta-\theta'\|^{\zeta}
    \end{split}
\end{equation*}
for a constant $C_2$ where we employed assumption \hyperref[assump:A7]{(A7)} in the last line. Thus the sum in \eqref{eq:collapsing_sum_2} is bounded by $CV_l^2(z)\|\theta-\theta'\|^{\zeta}$. In the case where $F_0\not=N$ the previous bound becomes
\begin{equation*}
    \begin{split}
        &C_1 \int \left(\prod_{j=1}^{N-1-i}\mu_{\theta}(dx_0^{j})\right)|\mu_{\theta}(dx_0^{N-i}) - \mu_{\theta'}(dx_0^{N-i})|\left(\prod_{j=N-i+1}^{N-1}\mu_{\theta'}(dx_0^{j})\right)
        \mathcal{V}_l(x_0^{F_0})\\
        =\;& \mathbbm{1}_{\{F_0=N-i\}}\int |\mu_{\theta}(dx_0^{N-i}) - \mu_{\theta'}(dx_0^{N-i})|
        \mathcal{V}_l(x_0^{N-i})+\mathbbm{1}_{\{F_0\not=N-i\}}\int |\mu_{\theta}(dx_0^{N-i}) - \mu_{\theta'}(dx_0^{N-i})|\\
        \leq \;& \int |\mu_{\theta}(dx_0^{N-i}) - \mu_{\theta'}(dx_0^{N-i})|
        \mathcal{V}_l(x_0^{N-i})\\ 
        \leq\; & C\|\theta-\theta'\|^{\zeta}.
    \end{split}
\end{equation*}
The last term of equation \eqref{eq:collapsing_sum_1} can be bounded in a similar manner to the second term. Finally because the set $\Theta$ is bounded the choice $\beta=\min(1,\zeta)$ works.
\end{proof}



\begin{thebibliography}{99}


\bibitem{andrieu} 
\textsc{Andrieu}, C., \textsc{Doucet}, A. \& \textsc{Holenstein},
R.~(2010). Particle Markov chain Monte Carlo methods (with discussion).
\textit{J. R. Statist. Soc. Ser. B}, \textbf{72}, 269--342.

\bibitem{andrieu2012}
{\sc Andrieu}, C., {\sc Lee}, A. \& {\sc Vihola}, M. ~(2013). Uniform Ergodicity of the Iterated Conditional SMC and Geometric Ergodicity of Particle Gibbs samplers. \emph{Bernoulli} {\bf 24}, 842-872.



\bibitem{andr3} 
{\sc Andrieu}, C., {\sc Moulines}, E. \& {\sc Priouret}, P.~(2005)
Stability of stochastic approximation under verifiable conditions. 
\emph{SIAM J. Control Optim.}, \textbf{44}, 283--312.




\bibitem{andr}
{\sc Andrieu}, C., \& {\sc Vihola}, M.~(2014).
Markovian stochastic approximation with expanding projections.
\emph{Bernoulli}, {\bf 20}, 545--585.


\bibitem{beskos}
{\sc Beskos}, A., {\sc Papaspiliopoulos}, O., {\sc Roberts}, G. O. \& {\sc Fearnhead}, P.~(2006)
Exact and computationally efficient likelihood-based estimation for discretely observed diffusion processes. 
\emph{J. R. Stat. Soc. Ser. B}, \textbf{68}, pp. 333--382.



%
%


\bibitem{cappe}
{\sc Capp\'e}, O., {\sc Ryden}, T, \& {\sc Moulines}, \'E.~(2005). \emph{Inference
in Hidden Markov Models}. Springer: New York.

\bibitem{kangaroo}
{\sc Caughley}, G., {\sc Shepherd}, N., \& {\sc Short}, J.~(1987)
\emph{Kangaroos: their ecology and management in the sheep rangelands
of Australia.} Cambridge University Press.




\bibitem{dennis}
{\sc Dennis}, B. \&  {\sc Costantino}, R. F.~(1988)
 Analysis of steady-state populations with the gamma abundance model:
application to Tribolium. \emph{Ecology}, {\bf 69}, 1200--1213.



\bibitem{Der}
{\sc  Dereich}, S. \& {\sc Muller-Gronbach}, T.~(2019)
General multilevel adaptations for stochastic approximation algorithms of 
Robbins-Monro and Polyal-Rupert type.
\emph{Numer. Math.}, \textbf{142}, 279--328.

\bibitem{delm}
{\sc Del Moral}, P., {\sc Horton}, E. \& {\sc Jasra}, A.~(2023). On the stability of positive semigroups. \emph{Ann. Appl. Probab.} (to appear).



%

\bibitem{frika}
{\sc Frikha}, N.~(2016). Multilevel stochastic approximation algorithms.
\emph{Ann. Appl. Probab.}, {\bf 26}, 933--985.




%


\bibitem{golightly}
{\sc Golightly}, A. \& {\sc Wilkinson}, D.~(2008). Bayesian inference for nonlinear multivariate diffusion models observed with error. \emph{Comp. Stat. Data Anal.}, {\bf 52}, 1674-1693.


\bibitem{graham}
{\sc Graham}, M. M., {\sc Thiery}, A. H. \& {\sc Beskos}, A.~(2022). 
Manifold Markov chain Monte Carlo methods for Bayesian inference in diffusion models. 
\emph{J. R. Stat. Soc. Ser. B}, {\bf 84}, 1229,-1256.



\bibitem{ub_grad}
{\sc Heng}, J., {\sc Houssineau}, J. \& {\sc Jasra}, A.~(2021). On unbiased score estimation for partially observed diffusions. arXiv preprint.


\bibitem{disc_model}
{\sc Heng}, J., {\sc Jasra}, A. {\sc Law}, K. J. H. \& {\sc Tarakanov}, A.~(2023). On unbiased estimation for discretized models. \emph{SIAM/ASA JUQ}, {\bf 11}, 616--645. 

\bibitem{mlpf}
{\sc Jasra}, A., {\sc Kamatani}, K., {\sc Law} K. J. H. \& {\sc Zhou}, Y.~(2017). 
Multilevel particle filters. \emph{SIAM J. Numer. Anal.}, {\bf 55}, 3068-3096.

\bibitem{par_uq}
{\sc Jasra}, A., {\sc Law}, K. J. H. \& {\sc Lu}, D.~(2022).
Unbiased estimation of the gradient of the log-likelihood in inverse problems. \emph{Stat. Comp.}, {\bf 31}, article 21.

\bibitem{ub_pf}
{\sc Jasra}, A., {\sc Law}, K. J. H. \& {\sc Yu}, F.~(2022). Unbiased filtering of a class of partially observed diffusions.
\emph{Adv. Appl. Probab.} {\bf 54}, 661-687.


\bibitem{knape}
{\sc Knape}, K. \& {\sc De Valpine}, P.~(2012).
 Fitting complex population models by combining particle filters with Markov chain Monte Carlo. 
 \emph{Ecology}, {\bf 93}, 256--263.
 
 \bibitem{kont}
 {\sc Kontoyiannis}, I. \& {\sc Meyn}, S. P.~(2003). Spectral theory and limit theorems for geometrically ergodic Markov processes.
 \emph{Ann. Appl. Probab.}, {\bf 13}, 304--362.


%

\bibitem{rhee}
{\sc Rhee}, C. H. \& {\sc Glynn}, P.~(2015). Unbiased estimation with square root convergence for SDE models. \emph{Op. Res.},~{\bf 63}, 1026--1043. 


\bibitem{robbins}
{\sc Robbins}, H. \& {\sc Monro}, H.~(1951). 
A stochastic approximation method. \emph{Ann. Math. Statist.}, \textbf{22}, 400--407.

\bibitem{sfs}
{\sc Ruzayqat, H., Beskos, A., Crisan, D., Jasra, A. \& Kantas, N.}.~(2023). Unbiased estimation using a class of diffusion processes. \emph{J. Comput. Phys.}, {\bf 472}(1), 111643.

\bibitem{matti}
{\sc Vihola}, M. (2018). Unbiased estimators and multilevel Monte Carlo. \emph{Op. Res.}, {\bf 66}, 448--462.

\bibitem{whiteley}
{\sc Whiteley}, N., {\sc Kantas}, N. \& {\sc Jasra}, A.~(2012).
Linear variance bounds for particle approximations of time-homogeneous Feynman-Kac formulae. \emph{Stoch. Proc. Appl.}, {\bf 122}, 1840--1865.




\end{thebibliography}
\end{document}